\documentclass[11pt]{article}

\usepackage[utf8]{inputenc} 
\usepackage[T1]{fontenc}    
\usepackage{hyperref}       
\usepackage{url}            
\usepackage{booktabs}       
\usepackage{amsthm}
\usepackage{amsmath}
\usepackage{amsfonts}       
\usepackage{nicefrac}       
\usepackage{microtype}      
\usepackage{subfigure}
\usepackage{graphicx}
\usepackage{bm}
\usepackage{bbm}
\usepackage{xspace}
\usepackage{fullpage}

\usepackage{color-edits}
\addauthor{xh}{green}
\addauthor{dk}{red}
\addauthor{yl}{blue}

\newcommand{\NodeSet}{\ensuremath{V}\xspace}
\newcommand{\EdgeSet}{\ensuremath{E}\xspace}
\newcommand{\Graph}{\ensuremath{G}\xspace}
\newcommand{\NodeNum}{\ensuremath{n}\xspace}
\newcommand{\EdgeNum}{\ensuremath{m}\xspace}
\newcommand{\EdgeW}[1][]{\ensuremath{\ifthenelse{\equal{#1}{}}{w}{w_{#1}}}\xspace}
\newcommand{\EdgeWV}{\ensuremath{\bm{w}}\xspace}
\newcommand{\Edge}{\ensuremath{e}\xspace}

\newcommand{\Delay}[2]{\ensuremath{d_{#1#2}}\xspace}
\newcommand{\CascNum}{\ensuremath{M}\xspace}
\newcommand{\Casc}[1][]{\ensuremath{\ifthenelse{\equal{#1}{}}{C}{C_{#1}}}\xspace}
\newcommand{\ActSet}[1][]{\ensuremath{\ifthenelse{\equal{#1}{}}{A}{A_{#1}}}\xspace}
\newcommand{\ActSetP}[1][]{\ensuremath{\ifthenelse{\equal{#1}{}}{A'}{A'_{#1}}}\xspace}
\newcommand{\SeedSet}[1][]{\ensuremath{\ifthenelse{\equal{#1}{}}{S}{S_{#1}}}\xspace}
\newcommand{\Cascs}{\ensuremath{\mathcal{C}}\xspace}
\newcommand{\ObvWin}{\ensuremath{\tau}\xspace} 

\newcommand{\InfFunc}[1][]{\ensuremath{\ifthenelse{\equal{#1}{}}{\bm{F}}{F_{#1}}}\xspace}
\newcommand{\InfFuncArg}[2][]{\ensuremath{\InfFunc[#1](#2)}\xspace}

\newcommand{\InfFuncD}[2]{\ensuremath{F_{#1}(#2)}\xspace}
\newcommand{\InfFuncV}[1]{\ensuremath{\bm{F}(#1)}\xspace}
\newcommand{\InfFuncDU}[3]{\ensuremath{F_{#1}^{#2}(#3)}\xspace}
\newcommand{\InfFuncDUO}[2]{\ensuremath{F_{#1}^{#2}}\xspace}

\newcommand{\NewInfFunc}{\ensuremath{\hat{\InfFunc}}\xspace}
\newcommand{\NewInfFuncV}[1]{\ensuremath{\hat{\InfFunc}(#1)}\xspace}
\newcommand{\NewInfFuncD}[2]{\ensuremath{\hat{F}_{#1}(#2)}\xspace}
\newcommand{\NewInfFuncDU}[3]{\ensuremath{\hat{F}_{#1}^{#2}(#3)}\xspace}

\newcommand{\TotalInfFunc}[1]{\ensuremath{\sigma(#1)}\xspace}

\newcommand{\ACTDIST}[1]{\ensuremath{\Delta_{#1}}\xspace}

\newcommand{\RetRate}{\ensuremath{r}\xspace}
\newcommand{\RetRateCenter}{\ensuremath{\bar{r}}\xspace}
\newcommand{\Transform}[1]{\hat{#1}}
\newcommand{\TransNodeSet}{\ensuremath{\Transform{\NodeSet}}\xspace}
\newcommand{\TransEdgeSet}{\ensuremath{\Transform{\EdgeSet}}\xspace}
\newcommand{\TransGraph}{\ensuremath{\Transform{\Graph}}\xspace}
\newcommand{\TransEdgeW}{\ensuremath{\Transform{\EdgeW}}\xspace}
\newcommand{\TransEdgeWD}[2]{\ensuremath{\EdgeSub{\TransEdgeW}{#1}{#2}}\xspace}

\newcommand{\NewNodeSet}{\ensuremath{\NodeSet'}\xspace}

\newcommand{\FeatureConst}{\ensuremath{C}\xspace}
\newcommand{\TrunConst}{\ensuremath{\lambda}\xspace}

\newcommand{\FeatureNum}{\ensuremath{K}\xspace}

\newcommand{\FeatureVD}[1]{\ensuremath{\bm{r}_{#1}}\xspace}
\newcommand{\Param}{\ensuremath{\beta}\xspace}
\newcommand{\ParamD}[1]{\ensuremath{\Param_{#1}}\xspace}
\newcommand{\ParamV}{\ensuremath{\bm{\beta}}\xspace}
\newcommand{\ParamGTD}[1]{\ensuremath{\Param^*_{#1}}\xspace}

\newcommand{\BaseFunc}[1][]{\ensuremath{\ifthenelse{\equal{#1}{}}{\phi}{\phi_{#1}}}\xspace}
\newcommand{\BaseFuncArg}[2][]{\ensuremath{\BaseFunc[#1](#2)}\xspace}

\newcommand{\LEG}[1][]{\ensuremath{\ifthenelse{\equal{#1}{}}{H}{H_{#1}}}\xspace} 
\newcommand{\LEGDist}[1][]{\ensuremath{\ifthenelse{\equal{#1}{}}{\Gamma}{\gamma_{#1}}}\xspace}
\newcommand{\FeatureSet}{\ensuremath{\mathcal{T}}\xspace}

\newcommand{\Label}[1][]{\ensuremath{\ifthenelse{\equal{#1}{}}{y}{y_{#1}}}\xspace}
\newcommand{\IncomLabel}[1][]{\ensuremath{\ifthenelse{\equal{#1}{}}{\Incom{\Label}}{\Incom{\Label}_{#1}}}\xspace}

\newcommand{\SeedDist}{\ensuremath{\mathcal{P}}\xspace}
\newcommand{\Model}[1][]{\ensuremath{\ifthenelse{\equal{#1}{}}{\mathcal{M}}{\mathcal{M}_{#1}}}\xspace}
\newcommand{\ModelClass}{\ensuremath{\mathcal{\InfFunc}_{\Model}}\xspace}
\newcommand{\LearnClass}{\ensuremath{\mathcal{\InfFunc}_{\mathcal{L}}}\xspace}
\newcommand{\LossFunc}{\ensuremath{\ell}\xspace}
\newcommand{\IncomLossFunc}{\ensuremath{\Incom{\LossFunc}}\xspace}

\newcommand{\SqLossFunc}{\ensuremath{\ell_{\text{sq}}}\xspace}
\newcommand{\Rade}[2][]{\ensuremath{\ifthenelse{\equal{#1}{}}{\mathcal{R}(#2)}{\mathcal{R}_{#1}(#2)}}\xspace}
\newcommand{\conv}[1]{\ensuremath{\mbox{conv}(#1)}\xspace}

\newcommand{\RelativeInterval}{\ensuremath{\eta}\xspace}

\newcommand{\Err}{\ensuremath{\text{err}_{\text{sq}}}\xspace}
\newcommand{\IncomErr}{\ensuremath{\widehat{\text{err}}_{\text{sq}}}\xspace}
\newcommand{\ErrLog}{\ensuremath{\text{err}_{\text{log}}}\xspace}

\newcommand{\EdgeSub}[3]{\ensuremath{#1_{#2#3}}\xspace}
\newcommand{\EdgeWD}[2]{\ensuremath{\EdgeSub{\EdgeW}{#1}{#2}}\xspace}
\newcommand{\ActSetT}[1]{\ensuremath{\ActSet_{#1}}\xspace}

\newcommand{\Incom}[1]{\ensuremath{\tilde{#1}}\xspace}
\newcommand{\IncomCasc}[1][]{\ensuremath{\ifthenelse{\equal{#1}{}}{\Incom{\Casc}}{\Incom{\Casc}_{#1}}}\xspace}
\newcommand{\IncomCascs}{\ensuremath{\Incom{\Cascs}}\xspace}
\newcommand{\IncomActSet}[1][]{\ensuremath{\ifthenelse{\equal{#1}{}}{\Incom{\ActSet}}{\Incom{\ActSet}_{#1}}}\xspace}

\newcommand{\CharFuncV}[1]{\ensuremath{\bm{\chi}_{#1}}\xspace}
\newcommand{\CharFuncD}[2]{\ensuremath{{\chi_{#2}(#1)}}\xspace}

\providecommand{\Expect}[2][]{\ensuremath{%
\ifthenelse{\equal{#1}{}}{\mathbb{E}}{\mathbb{E}_{#1}}%
\left[#2\right]}\xspace}

\providecommand{\PROB}{\ensuremath{{\rm Prob}}\xspace}
\providecommand{\Prob}[2][]{\ensuremath{%
\ifthenelse{\equal{#1}{}}{\PROB[#2]}{\PROB_{#1}[#2]}}\xspace}
\providecommand{\argmin}{\text{argmin}}
\providecommand{\argmax}{\text{argmax}}

\title{Learning Influence Functions from Incomplete Observations}

\author{
  Xinran He\qquad Ke Xu\qquad David Kempe\qquad Yan Liu \\
  University of Southern California, Los Angeles, CA 90089 \\
  \texttt{\{xinranhe, xuk, dkempe, yanliu.cs\}@usc.edu} \\
}

\newtheorem{theorem}{Theorem}
\newtheorem{lemma}[theorem]{Lemma}

\begin{document}

\maketitle

\begin{abstract}
We study the problem of learning influence functions under incomplete
observations of node activations.
Incomplete observations are a major concern as most (online and
real-world) social networks are not fully observable.
We establish both proper and improper PAC learnability of influence
functions under randomly missing observations.
Proper PAC learnability under the Discrete-Time Linear Threshold (DLT)
and Discrete-Time Independent Cascade (DIC) models is established by
reducing incomplete observations to complete observations
in a modified graph.
Our improper PAC learnability result applies for the DLT and DIC
models as well as the Continuous-Time Independent Cascade (CIC) model. 
It is based on a parametrization in terms of reachability features,
and also gives rise to an efficient and practical heuristic.
Experiments on synthetic and real-world datasets demonstrate the
ability of our method to compensate even for a fairly large
fraction of missing observations.
\end{abstract}

\section{Introduction}
Many social phenomena, 
such as the spread of diseases, behaviors, technologies, or products,
can naturally be modeled as the diffusion of a contagion across a network. 
Owing to the potentially high social or economic value of
  accelerating or inhibiting such diffusions, the goal of
understanding the flow of information and predicting information
cascades has been an active area of research
\cite{kempe_maximizing_2003, gomez-rodriguez_uncovering_2011, goyal_learning_2010,
  myers_convexity_2010,amin2014learning, rosenfeld2016discriminative}. 
In this context, a key task is learning 
\emph{influence functions}: the functions mapping
sets of initial adopters to the individuals who will be
influenced (also called \emph{active})
by the end of the diffusion process~\cite{kempe_maximizing_2003}. 

Many methods have been developed to solve the influence function
learning problem~\cite{goyal_learning_2010,
  gomez-rodriguez_uncovering_2011, du_learning_2012,
  gomez-rodriguez_inferring_2012, du_influence_2014,
  narasimhan_learnability_2015, praneeth_learning_2012,
  yang_mixture_2013, zhou_learning_2013}. 
Most approaches are based on fitting the parameters of a
diffusion model based on observations, e.g., \cite{gomez-rodriguez_inferring_2012,
  gomez-rodriguez_uncovering_2011, praneeth_learning_2012,
  goyal_learning_2010, narasimhan_learnability_2015}. 
Recently, Du et al.~\cite{du_influence_2014} proposed a
\emph{model-free} approach to learn influence functions as coverage functions;
Narasimhan et al.~\cite{narasimhan_learnability_2015} establish 
proper PAC learnability of influence functions under 
several widely-used diffusion models.


All existing approaches rely on the assumption that
the observations in the training dataset are complete, 
complete, in the sense that all active nodes are observed as being active. 
However, this assumption fails to hold in virtually all practical
applications
\cite{myers_information_2012,QMS11,chierichetti_reconstructing_2011,sadikov_correcting_2011}.
For example, social media data are usually collected through crawlers
or acquired with public APIs provided by social media platforms, such
as Twitter or Facebook. Due to non-technical reasons and established
restrictions on the APIs, it is often impossible to obtain a complete
set of observations even for a short period of time. 
In turn, the existence of unobserved nodes, links, or
  activations may lead to a significant misestimation of the diffusion
  model's parameters
\cite{quang_modeling_2011, myers_information_2012}.

In this paper, we take a step towards addressing
 the problem of learning influence functions from incomplete
observations (Here specifically we mean missing activation in the observed cascades). 
Missing data are a complicated phenomenon, but to address it
meaningfully and rigorously, one must make at least \emph{some}
assumptions about the process resulting in the loss of data.
As a first step, we focus on \emph{random} loss of observations: 
for each activated node independently, the node's activation
is observed only with probability \RetRate, the \emph{retention rate},
and fails to be observed with probability $1-\RetRate$.
Random observation loss naturally occurs
when crawling data from social media, where rate restrictions 
are likely to affect all observations equally.

We establish both proper and improper PAC learnability of influence
functions under incomplete observations for two popular
diffusion models: the Discrete-Time Independent Cascade (DIC) and
Discrete-Time Linear Threshold (DLT) models.
The result is proved by interpreting the incomplete observations as 
complete observations in a transformed graph, 
In fact, randomly missing observations do not 
significantly increase the required sample complexity.

The PAC learnability result implies good sample complexity bounds for
the DIC and DLT models. 
However, even without missing observations,
proper PAC learnability of the CIC and other models appears to be more challenging.
Furthermore, the PAC learnability result does not lead to an \emph{efficient}
algorithm, as it involves marginalizing a large number of
hidden variables (one for each node not observed to be active).


Towards designing more practical algorithms and
obtaining learnability under a broader class of diffusion
models, we pursue improper learning approaches.
Concretely, we use the parameterization of Du et
al.~\cite{du_influence_2014} in terms of reachability basis functions,
and optimize a modified loss function suggested by Natarajan et al.
\cite{natarajan2013learning} to address incomplete observations.
We prove that the algorithm ensures improper PAC learning for
the DIC, DLT and Continuous-Time Independent Cascade (CIC) models.
Experimental results on synthetic cascades generated from these 
diffusion models and real-world cascades in the MemeTracker dataset
demonstrate the effectiveness of our approach. 
Our algorithm achieves nearly a 20\% reduction in estimation error
compared to the best baseline methods on the MemeTracker dataset,
by compensating for incomplete observations.



Several recent works also aim to address the issue of missing
observations in social network analysis, but with different
emphases.
For example, Chierichetti et al.~\cite{chierichetti_reconstructing_2011}
and Sadikov et al.~\cite{sadikov_correcting_2011} mainly
focus on recovering the \emph{size} of a diffusion process,
while our task is to learn the influence functions from
several incomplete cascades.
Myers et al.~\cite{myers_information_2012} mainly
aim to model unobserved external influence in diffusion.
Duong et al.~\cite{QMS11} examine learning diffusion models with
missing links from \emph{complete} observations, while we learn
influence functions from incomplete cascades with missing
activations. 
Most related to our work are papers by Wu et
al.~\cite{Wu2013Parameter} and simultaneous work by 
Lokhov~\cite{Lokhov2016incomplete}. 
Both study the problem of network inference under incomplete observations. 
Lokhov proposes a dynamic message passing approach to 
marginalize all the missing activations,
in order to infer diffusion model parameters using maximum likelihood
estimation, while Wu et al.~develop an EM algorithm. 
Notice that the goal of learning the model parameters differs from our
goal of learning the influence functions directly.
Both \cite{Lokhov2016incomplete} and~\cite{Wu2013Parameter} provide
empirical evaluation, but do not provide theoretical guarantees.

\section{Preliminaries}
\subsection{Models of Diffusion and Incomplete Observations}
\paragraph{Diffusion Model.} 
We model propagation of opinions, products, or behaviors 
as a diffusion process over a social network. 
The social network is represented as a directed graph
$\Graph=(\NodeSet, \EdgeSet)$, where $\NodeNum = |\NodeSet|$ is the
number of nodes, and $\EdgeNum = |\EdgeSet|$ is the number of edges. 
Each edge $\Edge=(u,v)$ is associated with a parameter \EdgeWD{u}{v}
representing the strength of influence user $u$ has on $v$. 
We assume that the graph structure (the edge set $\EdgeSet$) is known,
while the parameters $\EdgeWD{u}{v}$ are to be learned.
Depending on the diffusion model, there are different ways to
represent the strength of influence between individuals.
Nodes can be in one of two states, \emph{inactive} or \emph{active}. 
We say that a node gets activated if it adopts the
opinion/product/behavior under the diffusion process. 
In this work, we focus on \emph{progressive} diffusion models, where a
node remains active once it gets activated.   

The diffusion process begins with a set of seed nodes (initial
adopters) $\SeedSet$, who start active.
The process then proceeds in discrete or continuous time: 
according to a probabilistic process, additional nodes may become
active based on the influence from their neighbors. 
Let $N(v)$ be the in-neighbors of node $v$ and \ActSetT{t} the set
of nodes activated by time $t$. 
We consider the following three widely used diffusion models:

\begin{itemize}
\item \textbf{Discrete-time Linear Threshold (DLT) model~\cite{kempe_maximizing_2003}:}  
Each node $v$ has a threshold $\theta_v$ drawn independently and
uniformly from the interval $[0,1]$. 
The diffusion under the DLT model unfolds in discrete time steps: 
a node $v$ becomes active at step $t$ if the total incoming weight from its
neighbors exceeds its threshold: 
$\sum_{u \in N(v) \cap \ActSetT{t-1}} \EdgeWD{u}{v} \geq \theta_v$.
\item \textbf{Discrete-time Independent Cascade (DIC) model~\cite{kempe_maximizing_2003}:} 
The DIC model is also a discrete-time model. 
Under the DIC model, the weight $\EdgeWD{u}{v}\in [0,1]$ captures an activation probability. 
When a node $u$ becomes active in step $t$, it attempts to activate
all currently inactive neighbors in step $t+1$. 
For each neighbor $v$, it succeeds with probability \EdgeWD{u}{v}.
If it succeeds, $v$ becomes active; otherwise, $v$ remains inactive. 
Once $u$ has made all these attempts, 
it does not get to make further activation attempts at later times.
\item \textbf{Continuous-time Independent Cascade (CIC) model~\cite{gomez-rodriguez_inferring_2012}:}
The CIC model unfolds in continuous time. 
Each edge $\Edge = (u,v)$ is associated with a delay distribution with
\EdgeWD{u}{v} as its parameter. 
When a node $u$ becomes newly active at time $t$, 
for every neighbor $v$ that is still inactive, 
a delay time \Delay{u}{v} is drawn from the delay distribution. 
\Delay{u}{v} is the duration it takes $u$ to activate $v$,
which could be infinite (if $u$ does not succeed in activating $v$). 
Nodes are considered activated by the process if they are activated
within a specified observation window $[0,\ObvWin]$.
\end{itemize}

Fix one of the diffusion models defined above and its parameters.
For each seed set \SeedSet, let \ACTDIST{\SeedSet} be the distribution
of final active sets when the seed set is \SeedSet.
(In the case of the DIC and DLT model, this is the set of active nodes
when no new activations occur; for the CIC model, it is the set of
nodes active at time \ObvWin.)
For any node $v$, let 
$\InfFuncArg[v]{\SeedSet} = \Prob[A \sim \ACTDIST{\SeedSet}]{v \in A}$ 
be the (marginal) probability that $v$ is activated according to the
dynamics of the diffusion model with initial seeds \SeedSet.
Then, define the \emph{influence function}
$\InfFunc: 2^\NodeSet \rightarrow [0,1]^{\NodeNum}$ mapping seed sets
to the vector of marginal activation probabilities:
$\InfFuncArg{\SeedSet} = [\InfFuncArg[1]{\SeedSet}, \ldots, \InfFuncArg[\NodeNum]{\SeedSet}]$.
Notice that the marginal probabilities do not capture the full
information about the diffusion process contained in 
\ACTDIST{\SeedSet} (since they do not observe
co-activation patterns), but they are sufficient for many
applications, such as influence maximization~\cite{kempe_maximizing_2003}
and influence estimation~\cite{du_scalable_2013}.


\paragraph{Cascades and Incomplete Observations.}
We focus on the problem of learning influence functions from cascades. 
A cascade $\Casc=(\SeedSet, \ActSet)$ is a realization of the random
diffusion process;
\SeedSet is the set of seeds and 
$\ActSet \sim \ACTDIST{\SeedSet}, \ActSet \supseteq \SeedSet$
is the set of activated nodes at the end of the random process.
Similar to Narasimhan et al.~\cite{narasimhan_learnability_2015}, 
we focus on
\emph{activation-only} observations, namely, we only observe
\emph{which nodes} were activated, but not \emph{when} these
activations occurred.\footnote{Narasimhan et al.~\cite{narasimhan_learnability_2015}
refer to this model as \emph{partial observations}; we change the
terminology to avoid confusion with ``incomplete observations.''}  

To capture the fact that some of the node activations may have
been unobserved, we use the following model of independently randomly
missing data:
for each (activated) node $v \in \ActSet \setminus \SeedSet$,
the activation of $v$ is actually \emph{observed}
  independently with probability \RetRate.
With probability $1-\RetRate$, the node's activation is unobservable.
For seed nodes $v \in \SeedSet$, the activation is never lost.
Formally, define \IncomActSet as follows: each $v \in \SeedSet$ is
deterministically in \IncomActSet, and each $v \in \ActSet \setminus \SeedSet$
is in \IncomActSet independently with probability \RetRate.
Then, the incomplete cascade is denoted by
$\IncomCasc = (\SeedSet, \IncomActSet)$.


\subsection{Objective Functions and Learning Goals}
To measure estimation error, 
we primarily use a quadratic loss function,
as in~\cite{narasimhan_learnability_2015, du_influence_2014}. 
For two $n$-dimensional vectors $\bm{x}, \bm{y}$, 
the quadratic loss is defined as 
$\SqLossFunc(\bm{x},\bm{y}) = \frac{1}{n} \cdot ||\bm{x} - \bm{y}||_2^2$.
We also use this notation when one or both
arguments are sets: when an argument is a set $S$, we formally mean
to use the \emph{indicator function} $\CharFuncV{S}$ as a vector,
where $\CharFuncD{v}{S} = 1$ if $v \in S$, 
and $\CharFuncD{v}{S} = 0$ otherwise. 
In particular, for an activated set \ActSet, we write
$\SqLossFunc(\ActSet, \InfFuncV{\SeedSet})
= \frac{1}{\NodeNum} ||\CharFuncV{\ActSet} - \InfFuncV{\SeedSet}||_2^2$.

We now formally define the problem of learning influence functions
from incomplete observations.
Let \SeedDist be a distribution over seed sets (i.e., a distribution
over $2^{\NodeSet}$), and fix a diffusion model \Model and parameters,
together giving rise to a distribution \ACTDIST{\SeedSet} for each
seed set.
The algorithm is given a set of \CascNum incomplete cascades
$\IncomCascs = \{(\SeedSet[1], \IncomActSet[1]), \ldots,
(\SeedSet[\CascNum], \IncomActSet[\CascNum])\}$,
where each \SeedSet[i] is drawn independently from \SeedDist,
and \IncomActSet[i] is obtained by the incomplete observation
  process described above from the (random) activated set 
$\ActSet[i] \sim \ACTDIST{\SeedSet[i]}$.
The goal is to learn an influence function \InfFunc that
accurately captures the diffusion process. 
Accuracy of the learned influence function is measured in terms of the
squared error with respect to the true model: 
$\Err[\InfFunc] =
\Expect[\SeedSet \sim \SeedDist, \ActSet \sim \ACTDIST{\SeedSet}]{
\SqLossFunc(\ActSet, \InfFuncV{\SeedSet})}$.
That is, the expectation is over the seed set 
and the randomness in the diffusion process,
but not the data loss process.

\paragraph{PAC Learnability of Influence Functions.}
We characterize the learnability of influence functions under
incomplete observations using the Probably Approximately Correct (PAC)
learning framework~\cite{valiant1984theory}. 
Let \ModelClass be the class of influence functions under the
diffusion model \Model,
and \LearnClass the class of influence functions the learning
  algorithm is allowed to choose from. 
We say that \ModelClass is PAC learnable if there exists an algorithm
$\mathcal{A}$ with the following property:
for all $\varepsilon, \delta\in (0,1)$, 
all parametrizations of the diffusion model,
and all distributions \SeedDist over seed sets \SeedSet: 
when given \emph{activation-only} and \emph{incomplete} training cascades
$\IncomCascs = \{ (\SeedSet[1], \IncomActSet[1]), \ldots,
(\SeedSet[\CascNum], \IncomActSet[\CascNum]) \}$ 
with $\CascNum \geq poly(\NodeNum, \EdgeNum, 1/\varepsilon, 1/\delta)$, 
$\mathcal{A}$ outputs an influence function $\InfFunc \in \LearnClass$ satisfying: 
$$
\Prob[\IncomCascs]{\Err[\InfFunc] - \Err[\InfFunc^*] \geq \varepsilon} \leq \delta.
$$

Here, $\InfFunc^* \in \ModelClass$
is the ground truth influence function.
The probability is over the training cascades, 
including the seed set generation, the stochastic diffusion process,
and the missing data process. 
We say that an influence function learning algorithm $\mathcal{A}$ is
\emph{proper} if $\LearnClass \subseteq \ModelClass$;
that is, the learned influence function is guaranteed to be an instance
of the true diffusion model. 
Otherwise, we say that $\mathcal{A}$ is an \emph{improper} learning algorithm. 
 
\section{Proper PAC Learning under Incomplete Observations}\label{Sec:theory}
In this section, we establish proper PAC learnability of influence
functions under 
the DIC and DLT models. 
For both diffusion models, $\ModelClass$ can be
parameterized by an edge parameter vector $\EdgeWV$, 
whose entries $\EdgeW[e]$
are the activation probabilities (DIC model) 
or edge weights (DLT model).
Our goal is to find an influence function 
$\InfFunc^{\EdgeWV}\in\ModelClass$ that outputs accurate marginal
activation probabilities.
While our goal is \emph{proper} learning --- meaning that the
function must be from $\ModelClass$ --- we do \emph{not} require
that the inferred parameters match the true edge parameters $\EdgeWV$. 
Our main theoretical results are summarized in Theorem~\ref{Thm:DIC}
and Theorem~\ref{Thm:DLT}.

\begin{theorem}\label{Thm:DIC}
Let $\lambda\in(0, 0.5)$. 
The class of influence functions under the DIC model 
in which all edge activation probabilities satisfy
$\EdgeW[e]\in [\lambda, 1-\lambda]$ is PAC learnable
under incomplete observations with retention rate $\RetRate$.
The sample complexity is 
$\tilde{O}(\frac{\NodeNum^3 \EdgeNum}{\varepsilon^2\RetRate^4})$. 
\end{theorem}

\begin{theorem}\label{Thm:DLT}
Let $\lambda\in(0, 0.5)$, and consider
the class of influence functions under the DLT model 
such that the edge weight for every edge satisfies
$\EdgeW[e]\in [\lambda, 1-\lambda]$,
and for every node $v$, 
$1 - \sum_{u\in N(v)}\EdgeWD{u}{v} \in [\lambda, 1-\lambda]$.
This class is PAC learnable under incomplete observations with
retention rate $\RetRate$.
The sample complexity is 
$\tilde{O}(\frac{\NodeNum^3 \EdgeNum}{\varepsilon^2\RetRate^4})$. 
\end{theorem}


In this section, we present the intuition and a proof sketch for
the two theorems.
Details of the proof are provided in Appendix~\ref{app:theory-proofs}.

The key idea of the proof of both theorems is that a set of incomplete
cascades $\IncomCascs$ on $\Graph$ under the two models can be
considered as essentially complete cascades on a transformed 
graph $\TransGraph=(\TransNodeSet, \TransEdgeSet)$. 
The influence functions of nodes in $\TransGraph$ can be learned using
a modification of the result of 
Narasimhan et al.~\cite{narasimhan_learnability_2015}.
Subsequently, the influence functions for $\Graph$ are
directly obtained from the influence functions for $\TransGraph$,
by exploiting that influence functions only focus on the
marginal activation probabilities.

\begin{figure}
\center
\includegraphics[width=0.6\textwidth]{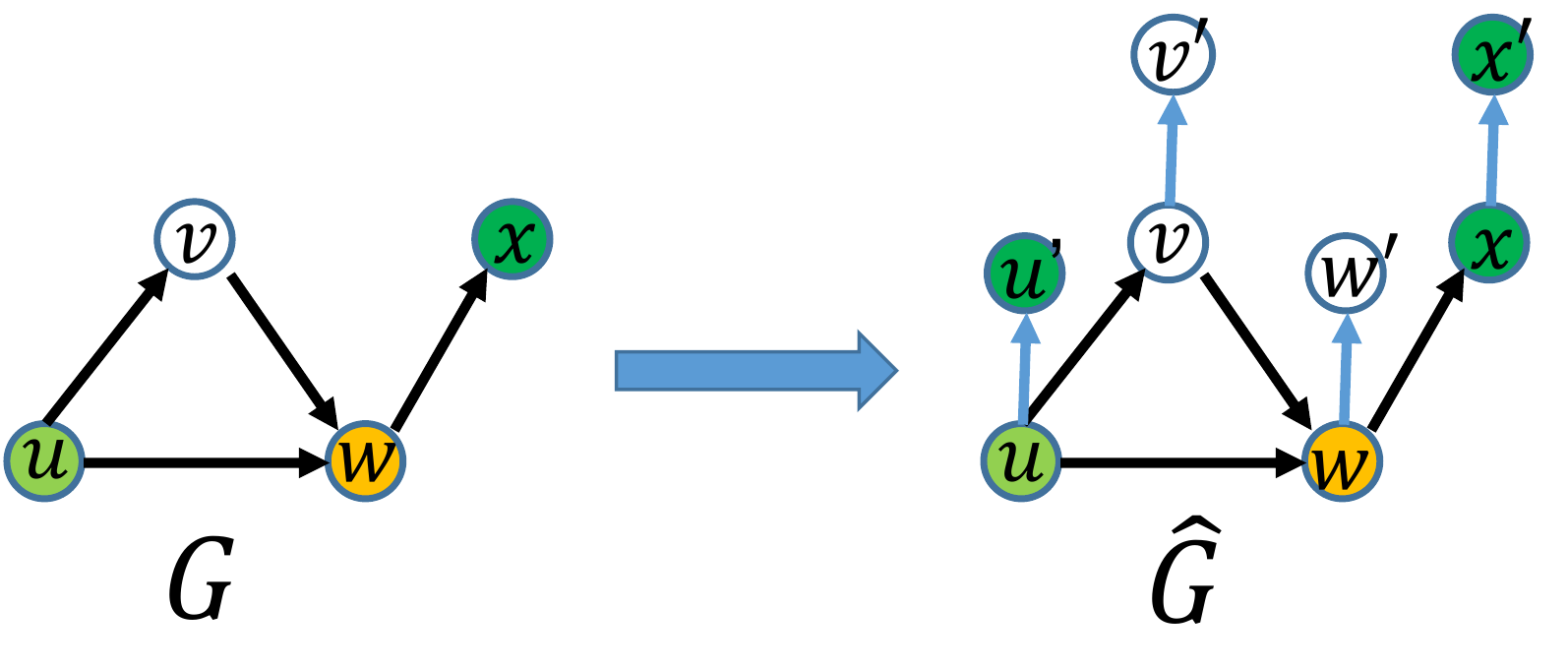} 
\caption{Illustration of the graph transformation. The light green
  node is the seed, the dark green nodes are the activated and
  observed nodes while the yellow node is activated but lost due to
  incomplete observations.}\label{Fig:graph_transform}
\end{figure}

The transformed graph $\TransGraph$ is built by adding a layer of
$\NodeNum$ nodes to the original graph $\Graph$. 
For each node $v \in \NodeSet$ of the original graph, 
we add a new node $v' \in \NewNodeSet$ and a directed edge $(v,v')$
with known and fixed edge parameter
$\TransEdgeWD{v}{v'} = \RetRate$.
(Conveniently, the same parameter value serves as activation
probability under the DIC model and as edge weight under the DLT model.)
The new nodes $\NewNodeSet$ have no other incident edges,
and we retain all edges $e = (u,v) \in \EdgeSet$.
Inferring their parameters is the learning task.
An example of the transformation on a simple graph consisting of four
nodes is shown in Figure~\ref{Fig:graph_transform}.


For each observed (incomplete) cascade 
$(\SeedSet[i], \IncomActSet[i])$ on $\Graph$ 
(with $\SeedSet[i] \subseteq \IncomActSet[i]$), 
we produce an observed activation set $\ActSetP[i]$ as follows: 
(1) for each $v \in \IncomActSet[i] \setminus \SeedSet[i]$, 
we let $v' \in \ActSetP[i]$ deterministically;
(2) for each $v \in \SeedSet[i]$ independently, we include $v' \in \ActSetP[i]$
with probability \RetRate.
This defines the training cascades 
$\Transform{\Cascs}=\{(\SeedSet[i], \ActSetP[i])\}$.

Now consider any edge parameters $\EdgeWV$, applied to both $\Graph$
and the first layer of $\TransGraph$.
Let $\InfFuncV{\SeedSet}$ denote the influence function on $\Graph$, and
$\NewInfFuncV{\SeedSet} = 
[\NewInfFuncD{1'}{\SeedSet}, \ldots, \NewInfFuncD{\NodeNum'}{\SeedSet}]$ 
the influence function of the nodes in the added layer $\NewNodeSet$
of $\TransGraph$.
Then, by the transformation, we get that
\begin{eqnarray}
\NewInfFuncD{v'}{\SeedSet} & = & \RetRate \cdot \InfFuncD{v}{\SeedSet} \label{Equ:tranformation}
\end{eqnarray}
for all nodes $v \in \NodeSet$.
And by the definition of the observation loss process, we also have
that for all non-seed nodes $v \notin \SeedSet[i]$,
\begin{eqnarray}
\Prob{v \in \IncomActSet[i]} \; = \;
\RetRate \cdot \InfFuncD{v}{\SeedSet} \; = \;
\NewInfFuncD{v'}{\SeedSet}. \nonumber
\end{eqnarray}

While the cascades $\Transform{\Cascs}$ are not complete on all of
$\TransGraph$, in a precise sense, they provide complete information
on the activation of nodes in $\NewNodeSet$. 
In Appendix~\ref{app:theory-proofs}, we show that Theorem~2 of Narasimhan et
al.~\cite{narasimhan_learnability_2015} can be extended to provide
identical guarantees for learning $\NewInfFuncV{\SeedSet}$
from the modified observed cascades $\Transform{\Cascs}$.
For the DIC model, this is a straightforward modification of the
proof from \cite{narasimhan_learnability_2015}. 
For the DLT model, \cite{narasimhan_learnability_2015} had not
established PAC
learnability\footnote{\cite{narasimhan_learnability_2015} shows that
  the DLT model with \emph{fixed} thresholds is PAC learnable under
  complete cascades. We study the DLT model when the thresholds are
  uniformly distributed random variables.}, so we provide a separate proof.

Because the results of \cite{narasimhan_learnability_2015} and our
generalizations ensure \emph{proper} learning, they provide edge
weights $\EdgeWV$ between the nodes of $\NodeSet$. 
We use these exact same edge weights to define the learned influence
functions in $\Graph$.
Equation~\eqref{Equ:tranformation} then implies that the learned
influence functions on $\NodeSet$ satisfy
$\InfFuncD{v}{\SeedSet} = \frac{1}{\RetRate} \cdot \NewInfFuncD{v'}{\SeedSet}$.
The detailed analysis in Appendix~\ref{app:theory-proofs} shows that
the learning error only scales by a multiplicative factor
$\frac{1}{\RetRate^2}$.

The PAC learnability result shows that there is no 
information-theoretical obstacle to influence function learning under
incomplete observations. 
However, it does not imply an \emph{efficient} algorithm. 
The reason is that a hidden variable would be
associated with each node not observed to be active, 
and computing the objective function for empirical risk
minimization would require marginalizing over all of the hidden
variables. 
The proper PAC learnability result also does not readily generalize
to the CIC model and other diffusion models, even under complete observations. 
This is due to the lack of a succinct characterization
of influence functions as under the DIC and DLT models. 
Therefore, in the next section, we explore
improper learning approaches with the goal of
designing practical algorithms and establishing learnability under
a broader class of diffusion models.

\section{Efficient Improper Learning Algorithm}\label{Sec:algorithm}
In this section, we develop improper learning algorithms for
efficient influence function learning. Instead of parameterizing the
influence functions using the edge parameters, we adopt the model-free
influence function learning framework, InfluLearner, proposed by Du
et al.~\cite{du_influence_2014} to represent the influence function as
a sum of weighted basis functions. From now on, we focus on the
influence function $\InfFuncD{v}{\SeedSet}$ of a single fixed node $v$. 

\paragraph{Influence Function Parameterization.} For all three
diffusion models (CIC, DIC and DLT), the diffusion process can be
characterized equivalently using live-edge graphs.
Concretely, the results of \cite{kempe_maximizing_2003, du_scalable_2013}
state that for each instance of the CIC, DIC, and DLT models, 
there exists a distribution \LEGDist over live-edge graphs \LEG 
assigning probability \LEGDist[\LEG] to each live-edge graph \LEG
such that
$\InfFuncDU{v}{*}{\SeedSet} 
= \sum_{\LEG: \text{\emph{at least} one node in $S$ has a path to $v$ in \LEG}} \LEGDist[\LEG]$.

To reduce the representation complexity, notice that from the
perspective of activating $v$, two different live-edge graphs
$\LEG, \LEG'$ are ``equivalent'' if $v$ is reachable from exactly the
same nodes in $\LEG$ and $\LEG'$.
Therefore, for any node set $T$, let
$\ParamGTD{T} := \sum_{\LEG: \text{\emph{exactly} the nodes in $T$ have paths
  to $v$ in \LEG}} \LEGDist[\LEG]$.
We then use characteristic vectors as feature vectors
$\FeatureVD{T} = \CharFuncV{T}$, where we will interpret the entry for
node $u$ as $u$ having a path to $v$ in a live-edge graph.
More precisely, let $\BaseFuncArg{x} = \min\{x ,1\}$, 
and \CharFuncV{\SeedSet} the characteristic vector of the seed set \SeedSet. 
Then, 
$\BaseFuncArg{\CharFuncV{\SeedSet}^{\top}\cdot\FeatureVD{T}} = 1$
if and only if $v$ is reachable from \SeedSet, and we can write
$$
\InfFuncDU{v}{*}{\SeedSet} = \sum_{T} \ParamGTD{T}
\cdot \BaseFuncArg{\CharFuncV{\SeedSet}^{\top}\cdot\FeatureVD{T}}.
$$

This representation still has exponentially many features (one for each $T$).
In order to make the learning problem tractable, we sample a smaller
set \FeatureSet of \FeatureNum features from a suitably chosen distribution,
implicitly setting the weights \ParamD{T} of all other features to 0.
Thus, we will parametrize the learned influence function as
$$
\InfFuncDU{v}{\ParamV}{\SeedSet} = \sum_{T \in \FeatureSet} \ParamD{T}
\cdot \BaseFuncArg{\CharFuncV{\SeedSet}^{\top}\cdot\FeatureVD{T}}.
$$

The goal is then to learn weights \ParamD{T} for the sampled features.
(They will form a distribution, i.e., $||\ParamV||_1 = 1$ and $\ParamV\geq 0$.)
The crux of the analysis is to show that a sufficiently small number
\FeatureNum of features (i.e., sampled sets) suffices for a good
approximation, and that the weights can be learned efficiently from a
limited number of observed incomplete cascades.
Specifically, we consider the log likelihood function
$\LossFunc(t, \Label) = \Label \log t + (1-\Label)\log(1-t)$,
and learn the parameter vector \ParamV via the following maximum likelihood
estimation problem: 
\begin{center}
\begin{tabular}{cc}
Maximize & $\sum_{i=1}^{\CascNum} \LossFunc(\InfFuncDU{v}{\ParamV}{\SeedSet[i]}, \CharFuncD{v}{\ActSet[i]})$\\
subject to & $||\ParamV||_1 = 1, \ParamV \geq 0$.
\end{tabular}
\end{center}

\paragraph{Handling Incomplete Observations.} 
The maximum likelihood estimation cannot be directly applied
to incomplete cascades, as we do not have access to \ActSet[i]
(only the incomplete version \IncomActSet[i]). 
To address this issue, notice that
the MLE problem is actually a binary classification problem with
log loss and $\Label[i]=\CharFuncD{v}{\ActSet[i]}$ as the label. 
From this perspective, incompleteness is simply class-conditional
noise on the labels.
Let $\IncomLabel[i] = \CharFuncD{v}{\IncomActSet[i]}$ be our
\emph{observation} of whether $v$ was activated or not under the
incomplete cascade $i$. Then,
$$
\Prob{\IncomLabel[i] = 1 | \Label[i] = 1} = \RetRate\ \ 
\text{ and }\ \ 
\Prob{\IncomLabel[i] = 1 | \Label[i] = 0} = 0.
$$
In words, the incomplete observation \IncomLabel[i] suffers from
one-sided error compared to the complete observation \Label[i]. 
Known techniques can be used to address this issue.
By results of Natarajan et al.~\cite{natarajan2013learning},
we can construct an unbiased estimator of $\LossFunc(t, \Label)$ using
only the incomplete observations \IncomLabel, as in the following lemma. 

\begin{lemma}[Corollary of Lemma 1~\cite{natarajan2013learning}]\label{Lem:adjust_loss}
Let \Label be the true activation of node $v$ and \IncomLabel
the incomplete observation.
Then, defining
$$
\IncomLossFunc(t, \Label) := \frac{1}{\RetRate}\Label\log t 
                           + \frac{2\RetRate - 1}{\RetRate}(1-\Label) \log (1-t),
$$
for any $t$, we have 
$\Expect[\IncomLabel]{\IncomLossFunc(t, \IncomLabel)} = \LossFunc(t, \Label)$.
\end{lemma}

Based on this lemma, we solve the maximum likelihood
estimation problem with the adjusted likelihood function
$\IncomLossFunc(t, \Label)$:
\begin{eqnarray}
\mbox{Maximize} & 
\sum_{i=1}^{\CascNum} \IncomLossFunc(\InfFuncDU{v}{\ParamV}{\SeedSet[i]},
                                     \CharFuncD{v}{\IncomActSet[i]})
  \label{Eqn:log-optimization} \\
\mbox{subject to} & ||\ParamV||_1 = 1, \ParamV\geq 0. \nonumber
\end{eqnarray}

We analyze conditions under which the solution to
\eqref{Eqn:log-optimization} provides improper PAC learnability
under incomplete observations; these conditions will apply for 
all three diffusion models.

These conditions are similar to those of Lemma 1 in the work of Du et
al.~\cite{du_influence_2014}, and concern the approximability of 
the reachability distribution \ParamGTD{T}.
Specifically, let $q$ be a distribution over node sets $T$ such that
$q(T) \leq \FeatureConst \ParamGTD{T}$ for all node sets $T$. 
Here, \FeatureConst is a (possibly very large) number that we will
want to bound below.
Let $T_1, \ldots, T_{\FeatureNum}$ be \FeatureNum i.i.d.~samples drawn
from the distribution $q$. 
The features are then $\FeatureVD{k} = \CharFuncV{T_k}$. 
We use the truncated version of the function 
$\InfFuncDU{v}{\ParamV, \TrunConst}{\SeedSet}$ with parameter\footnote{%
The proof of Theorem~\ref{Thm:empirical} in Appendix~\ref{app:algorithm-proof}
will show how to choose $\TrunConst$.}
$\TrunConst$ as in~\cite{du_influence_2014}:
$$
\InfFuncDU{v}{\ParamV, \TrunConst}{\SeedSet} 
= (1 - 2\TrunConst) \InfFuncDU{v}{\ParamV}{\SeedSet} + \TrunConst.
$$

Let \Model[\TrunConst] be the class of all such truncated influence
functions, and 
$\InfFuncDUO{v}{\tilde{\ParamV}, \TrunConst} \in \Model[\TrunConst]$
the influence functions obtained from the optimization problem
\eqref{Eqn:log-optimization}.
The following theorem (proved in Appendix~\ref{app:algorithm-proof})
establishes the accuracy of the learned functions.

\begin{theorem} \label{Thm:empirical}
Assume that the learning algorithm uses
$\FeatureNum = \tilde{\Omega}(\frac{C^2}{\varepsilon^2})$
features in the influence function it constructs, and observes\footnote{%
The $\tilde{\Omega}$ notation suppresses all logarithmic terms except
$\log \FeatureConst$, 
as $\FeatureConst$ could be exponential or worse in the number of nodes.}
$\CascNum=\tilde{\Omega}(\frac{\log\FeatureConst}{\varepsilon^4\RetRate^2})$
incomplete cascades with retention rate $\RetRate$.
Then, with probability at least $1-\delta$, 
the learned influence functions \InfFuncDUO{v}{\tilde{\ParamV},\TrunConst}
for each node $v$ and seed distribution \SeedDist satisfy
$$
\Expect[\SeedSet \sim \SeedDist]{ 
     (\InfFuncDU{v}{\tilde{\ParamV},\TrunConst}{\SeedSet} 
      - \InfFuncDU{v}{*}{\SeedSet})^2 } \; \leq \; \varepsilon.
$$
\end{theorem}
The theorem implies that with enough incomplete cascades, an
algorithm can approximate the ground truth influence function to
arbitrary accuracy. 
Therefore, all three diffusion models are improperly PAC learnable
under incomplete observations. 
The final sample complexity does not contain the graph size, 
but is implicitly captured by \FeatureConst, 
which will depend on the graph and how well the distribution
\ParamGTD{T} can be approximated. 
Notice that with $\RetRate = 1$, our bound on \CascNum has
logarithmic dependency on \FeatureConst instead of polynomial, as in
\cite{du_influence_2014}. The reason for this improvement is discussed
further in Appendix~\ref{app:algorithm-proof}.

\paragraph{Efficient Implementation.} 
As mentioned above, the features $T$ cannot be sampled from
the exact reachability distribution \ParamGTD{T}, because it is
inaccessible (and complex).
In order to obtain useful guarantees from Theorem~\ref{Thm:empirical},
we follow the approach of Du et al.~\cite{du_influence_2014}, and
approximate the distribution \ParamGTD{T} with 
the product of the marginal distributions, estimated from
observed cascades.

The optimization problem~\eqref{Eqn:log-optimization} is
convex and can therefore be solved in time polynomial in the number of
features \FeatureNum.
However, the guarantees in Theorem~\ref{Thm:empirical} require
a possibly large number of features.
In order to obtain an \emph{efficient} algorithm for practical use and
our experiments, we sacrifice the guarantee and use a fixed number of features. 

Notice that the optimization problem~\eqref{Eqn:log-optimization} can
be solved independently for each node $v$; the learned functions 
\InfFuncArg[i]{\SeedSet} can then be combined into
$\InfFuncV{\SeedSet} = [\InfFuncD{1}{\SeedSet}, \ldots, \InfFuncD{n}{\SeedSet}]$.
As the optimization problem factorizes over nodes, the method is
obviously parallelizable, thus scaling to large networks. 

A further point regarding the implementation: in our theoretical
analysis, we assumed that the retention rate \RetRate is known to the
learning algorithm.
In practice, it can be estimated via cross-validation. 
As we show in the next section, the algorithm is not very sensitive to
the misspecification of the retention rate.

\section{Experiments} \label{Sec:experiments}
In this section, we experimentally evaluate the algorithm from
Section~\ref{Sec:algorithm}. Since no other methods explicitly
account for incomplete observations, we compare it to several
state-of-the-art methods for influence function learning with full
information.
Hence, the main goal of the comparison is 
to examine to what extent the impact of missing
data can be mitigated by being aware of it. 
We compare our algorithm to the following approaches:
(1) \textbf{CIC} is an approach fitting the parameters of a CIC model,
using the NetRate algorithm~\cite{gomez-rodriguez_uncovering_2011}
with exponential delay distribution.
(2) \textbf{DIC} fits the activation probabilities of a DIC model
using the method in~\cite{praneeth_learning_2012}.
(3) \textbf{InfluLearner} is the model-free approach proposed by Du et
al.~in \cite{du_influence_2014} and discussed in Section~\ref{Sec:algorithm}.
(4) \textbf{Logistic} uses logistic regression to learn the influence functions 
$\InfFuncD{u}{\SeedSet}=f(\CharFuncV{\SeedSet}^{\top}\cdot\bm{c}_u + b)$
for each $u$ independently, 
where $\bm{c}_u$ is a learnable weight vector and 
$f(x)=\frac{1}{1+e^{-x}}$ is the logistic function.
(5)\textbf{Linear} uses linear regression to learn the total influence
$\TotalInfFunc{\SeedSet}=\bm{c}^{\top}\cdot\CharFuncV{S}+ b$ of the set $S$.
Notice that the CIC and DIC methods have access to the 
\emph{activation time} of each node in addition to the final
activation status, giving them an inherent advantage.


\subsection{Synthetic cascades}
\paragraph{Data generation.} We generate synthetic networks with
core-peripheral structure following the Kronecker graph
model~\cite{leskovec_kronecker_2010} with parameter matrix $[0.9, 0.5;
0.5, 0.3]$.\footnote{We also experimented on Kronecker graphs with
  hierarchical community structure ($[0.9, 0.1; 0.1, 0.9]$) and random
  structure ($[0.5, 0.5; 0.5 ,0.5]$). The results are similar and
  omitted due to space constraints.} 
Each generated network has $512$ nodes and $1024$ edges. 

We then generate synthetic cascades following the CIC, DIC and DLT models. 
For the CIC model, we use an exponential delay distribution on each
edge whose parameters are drawn independently and uniformly from $[0,1]$. 
The observation window length is $\ObvWin = 1.0$. 
For the DIC model, the activation probability for each edge is chosen
independently and uniformly from $[0, 0.4]$. 
For the DLT model, we follow~\cite{kempe_maximizing_2003} and set the
edge weight \EdgeWD{u}{v} as $1/d_v$ where $d_v$ is the in-degree of
node $v$. 
For each model, we generate $8192$ cascades as training data. 
The seed sets are sampled uniformly at random with sizes drawn from a
power law distribution with parameter $2.5$. 
The generated cascades have average sizes of $10.8$, $12.8$ and $13.0$ in
the CIC, DIC and DLT models, respectively. 
We then create incomplete cascades by varying the retention rate
between $0.1$ and $0.9$.
The test set contains $200$ independently sampled seed sets generated
in the same way as the training data. 
To sidestep the computational cost of running Monte Carlo
simulations, we estimate the ground truth influence of the test seed sets
using the method proposed in~\cite{du_influence_2014}, with the true
model parameters. 

\paragraph{Algorithm settings.}
We apply all algorithms to cascades generated from all three
models; that is, we also consider the results under model misspecification.
Whenever applicable, we set the hyperparameters of the five comparison
algorithms to the ground truth values.
When applying the NetRate algorithm to discrete-time cascades, 
we set the observation window to $10.0$. 
When applying the method in~\cite{praneeth_learning_2012} to
continuous-time cascades, we round activation times up to the nearest
multiple of 0.1, resulting in 10 discrete time steps.
For the model-free approaches (InfluLearner and our algorithm),
we use $\FeatureNum = 200$ features. 

\paragraph{Results.} 
Figure~\ref{Fig:syn_results} shows the Mean Absolute Error (MAE)
between the estimated total influence \TotalInfFunc{\SeedSet} and 
the true influence value, averaged over all testing seed sets.
For each setting (diffusion model and retention rate), the reported MAE is
averaged over five independent runs. 
\begin{figure}
\begin{center}
\subfigure[CIC]{
\includegraphics[width=0.32\textwidth]{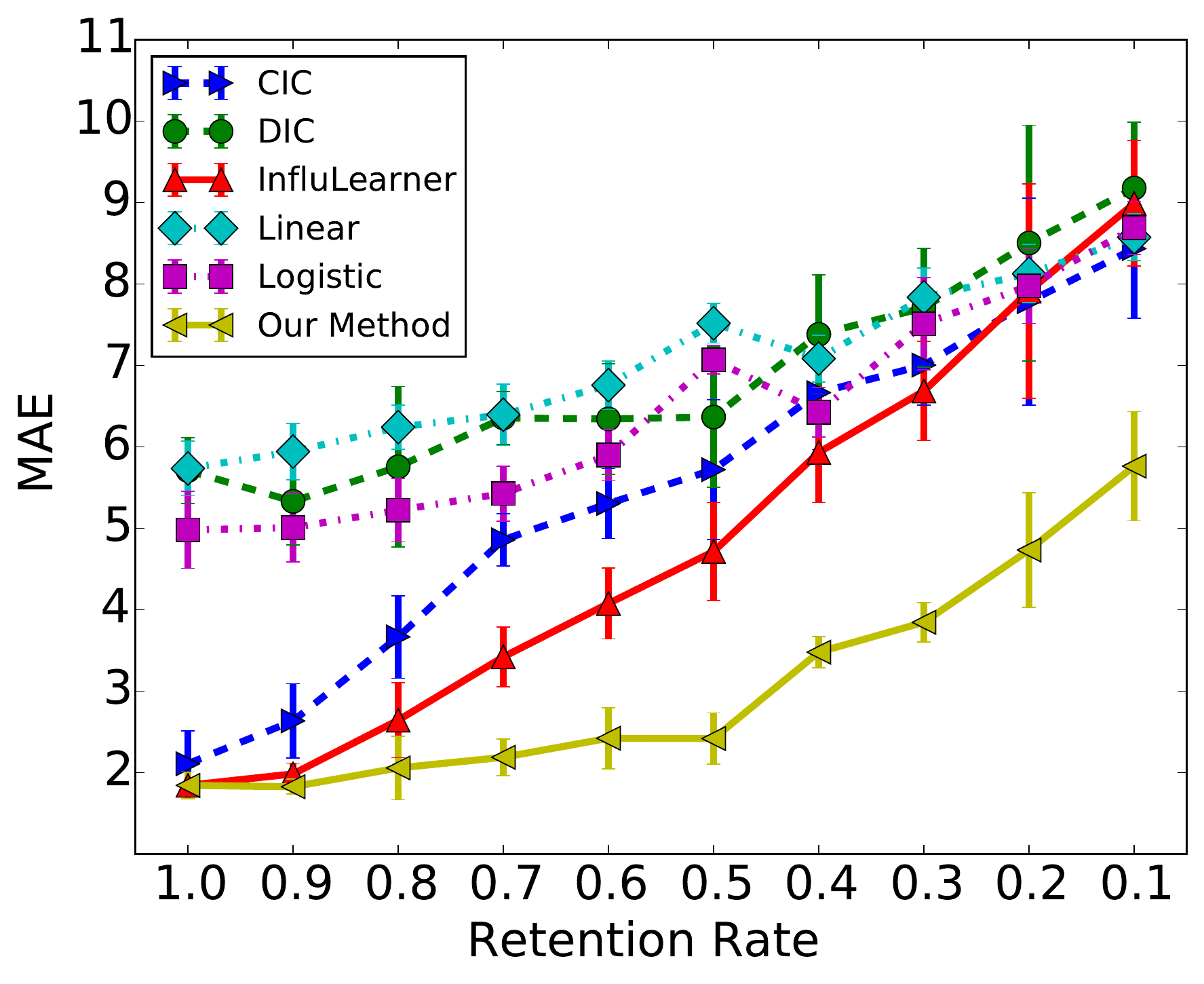}}
\subfigure[DIC]{
\includegraphics[width=0.32\textwidth]{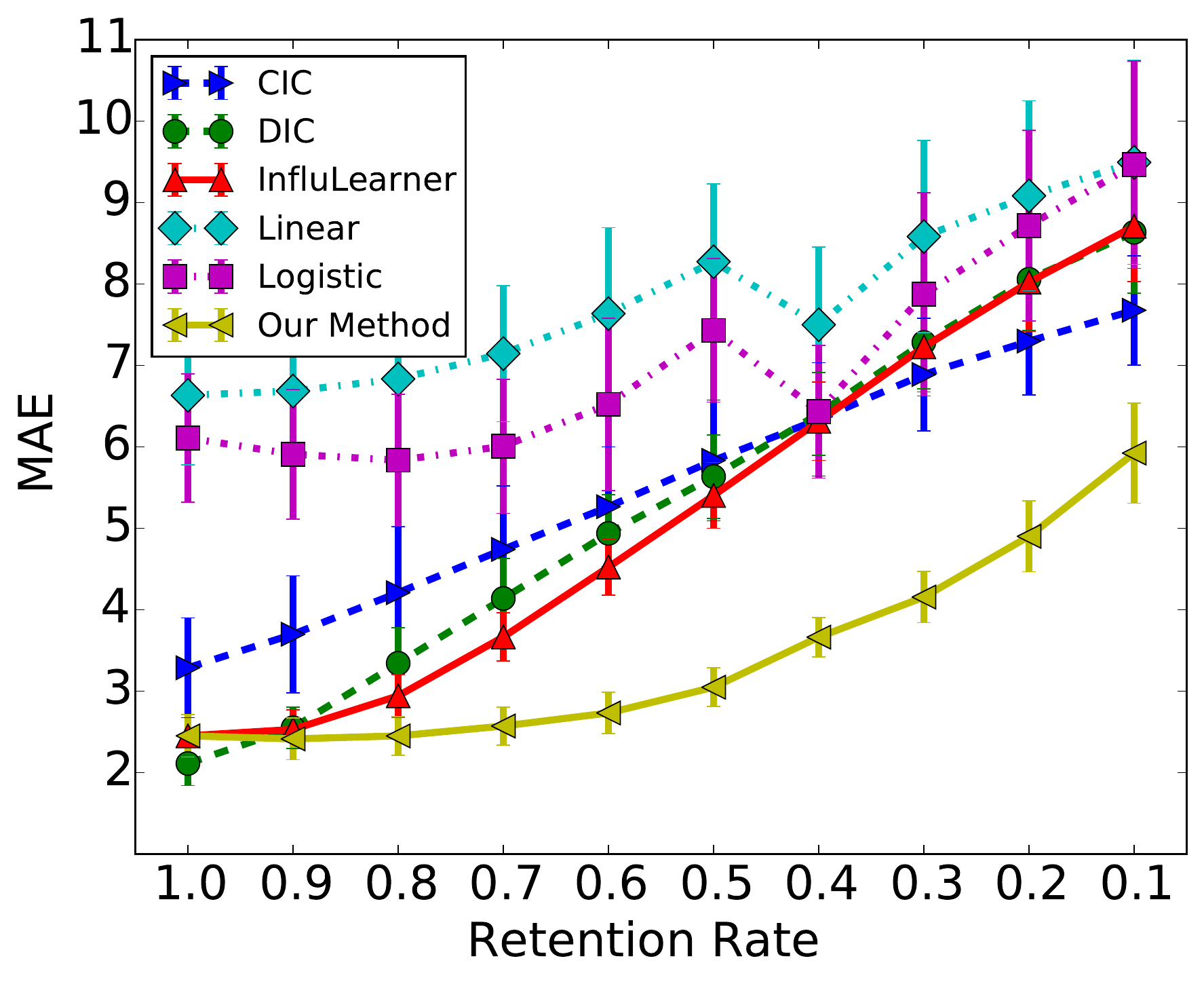}}
\subfigure[DLT]{
\includegraphics[width=0.32\textwidth]{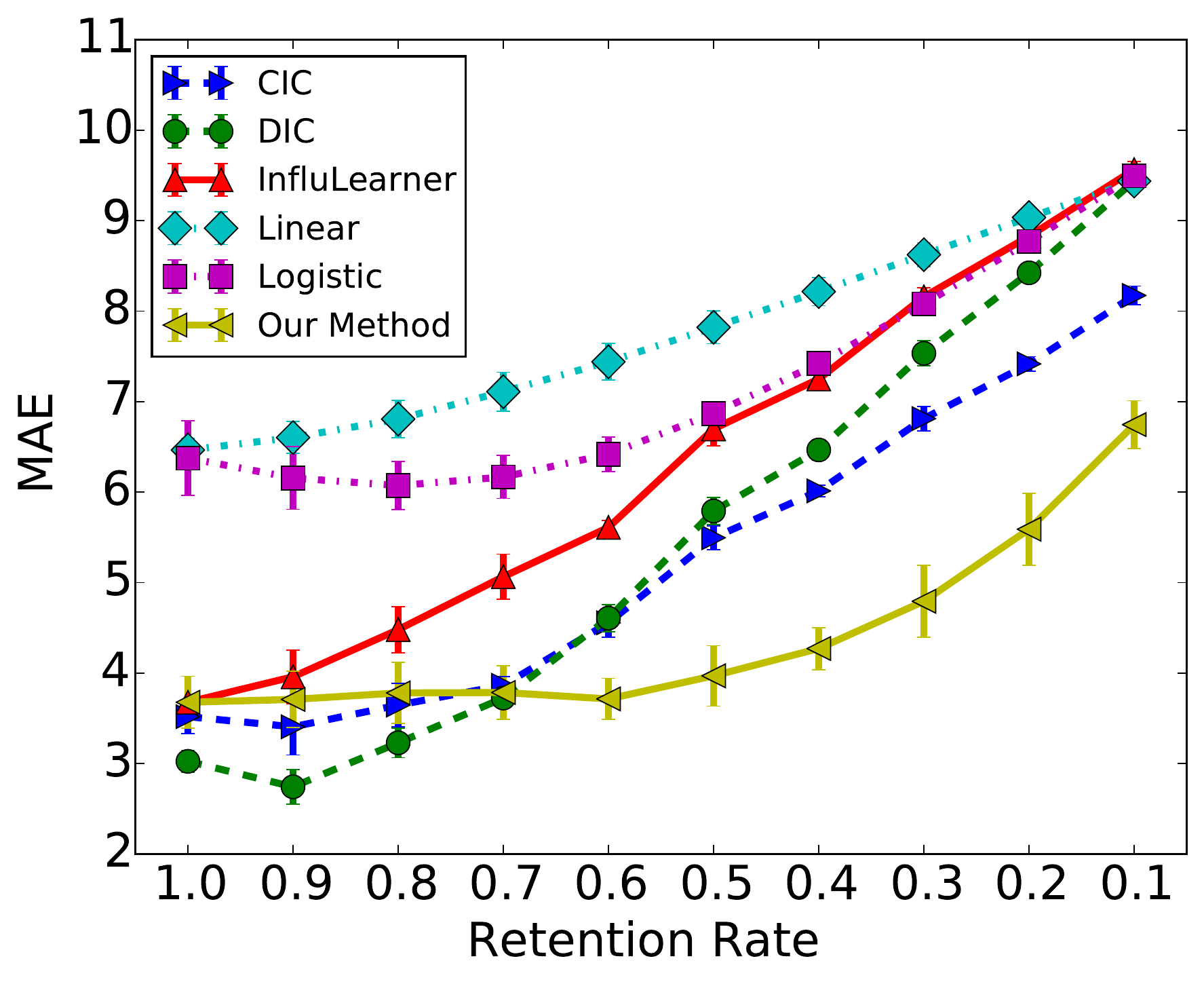}}
\end{center}
\caption{MAE of estimated influence as a function of the retention rate
  on synthetic datasets for (a) CIC model, (b) DIC model, (c) DLT model. 
  The error bars shows one standard deviation.\label{Fig:syn_results}}
\end{figure}

The main insight is that accounting for missing observations indeed
strongly mitigates their effect: notice that for retention
  rates as small as
$0.5$, our algorithm can almost completely compensate for the data
loss, whereas both the model-free and parameter fitting approaches
deteriorate significantly even for retention rates close to 1.
For the parameter fitting approaches, even such large retention rates
can lead to missing and spurious edges in the inferred networks, 
and thus significant estimation errors.
Additional observations include that fitting influence using (linear
or logistic) regression does not perform well at all, and that the CIC
inference approach appears more robust to model misspecification than
the DIC approach.


\paragraph{Sensitivity of retention rate.} 
We presented the algorithms as knowing \RetRate.
Since \RetRate itself is inferred from noisy data, it may be somewhat
misestimated.
Figure~\ref{Fig:uncertain} shows the impact of misestimating \RetRate.
We generate synthetic cascades from all three diffusion models with a
true retention rate of $0.8$, and then apply our algorithm with
(incorrect) retention rate 
$\RetRate \in \{0.6, 0.65, \ldots, 0.95, 1\}$. 
The results are averaged over five independent runs. 
While the performance decreases as the misestimation gets worse
(after all, with $\RetRate = 1$, the algorithm is basically the same as
InfluLearner), the degradation is graceful.

\paragraph{Non-uniform retention rate}
In practice, different nodes may have different retention rates, while we
may be able only to estimate the mean retention rate \RetRate. 
For our experiments, we draw each node $v$'s retention rate
$\RetRate_v$ independently from a distribution with mean \RetRate. 
Specifically, in our experiments, we use uniform and Gaussian
distributions;
for the uniform distribution, we draw
$\RetRate_v \sim \text{Unif}[\RetRate-\sigma, \RetRate+\sigma]$;
for the Gaussian distribution, we draw
$\RetRate_v \sim N(\RetRate, \sigma^2)$, truncating draws at 0 and
1.\footnote{The truncation could lead to a bias on the mean of
  $\RetRate$. However, empirical simulations show that the bias is
  negligible (only $0.01$ when $\sigma=0.2$).}
In both cases, $\sigma$ measures the level of noise in the estimated
retention rate. 
We set the mean retention rate $\RetRate$ to $0.8$ and vary
$\sigma$ in $\{0, 0.02, 0.05, 0.1, 0.2\}$. 
Figure~\ref{Fig:random} shows the results for the Gaussian
distribution; the results of uniform distribution are similar 
and omitted. 
The results show that our model is very robust to random and
independent perturbations of individual retention rates for each 
node.
\begin{figure}
\begin{center}
\includegraphics[width=0.5\textwidth]{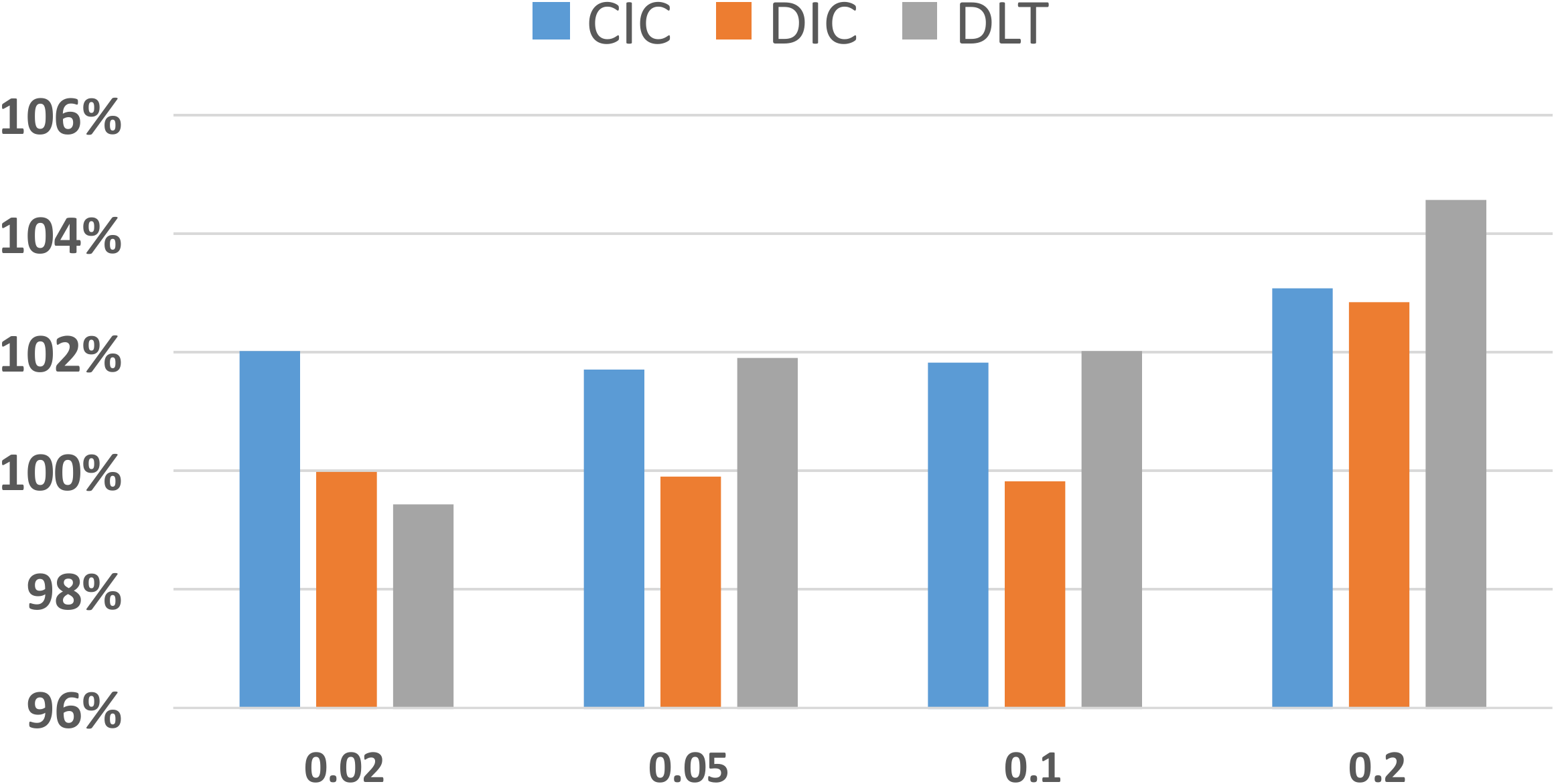}
\end{center}
\caption{Relative error in MAE when the true retention rates are drawn from
  the truncated Gaussian distribution. The $x$-axis shows standard deviation $\sigma$ of the retention
  rates from the mean, and the $y$-axis is the relative
  difference of the MAE compared to the case where all retention rates are
  the same and known. \label{Fig:random}}
\end{figure}

\subsection{Influence Estimation on real cascades}
We further evaluate the performance of our method on the real-world
MemeTracker\footnote{We use the preprocessed version of the
dataset released by Du et al.~\cite{du_influence_2014}
and available at
\url{http://www.cc.gatech.edu/~ndu8/InfluLearner.html}. 
Notice that the dataset is semi-real, as multi-node seed cascades are
artificially created by merging single-node seed cascades.}
dataset \cite{leskovec2009meme}.
The dataset consists of the propagation of short textual phrases,
referred to as \emph{Memes}, via the publication of blog posts and
main-stream media news articles between March 2011 and February 2012.
Specifically, the dataset contains seven groups of cascades
corresponding to the propagation of Memes with certain keywords, 
namely ``apple and jobs'', ``tsunami earthquake'', ``william kate marriage''', 
``occupy wall-street'', ``airstrikes'', ``egypt'' and ``elections.''
Each cascade group consists of $1000$
nodes, with a number of cascades varying from $1000$ to $44000$.
We follow exactly the same evaluation method as
Du et al.~\cite{du_influence_2014} with a training/test set split of
60\%/40\%.

\begin{figure}[!t]
\minipage[t]{0.47\textwidth}
  \includegraphics[width=\linewidth]{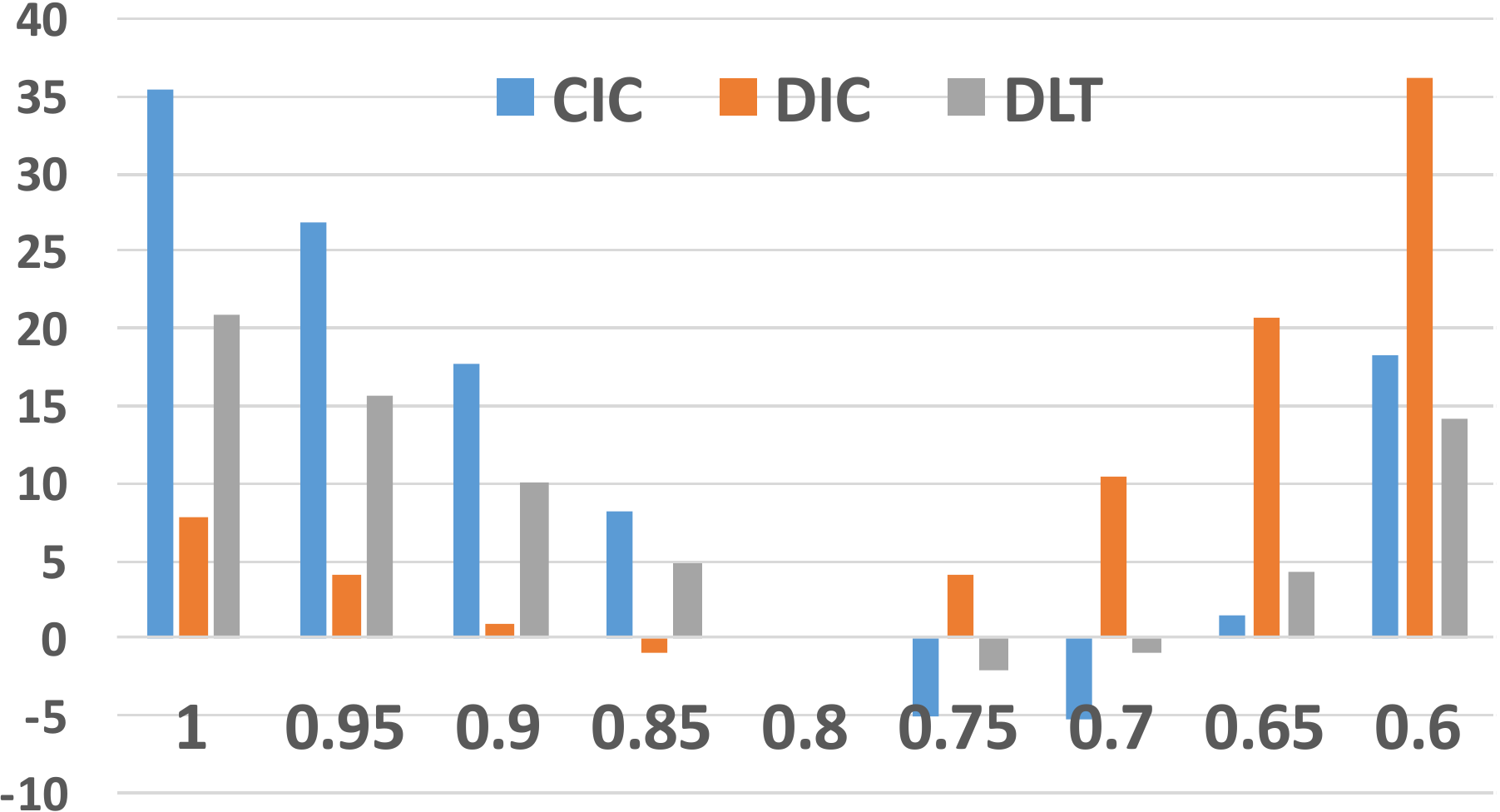}
  \caption{Relative error in MAE under retention rate misspecification. 
    $x$-axis: retention rate \RetRate used by the algorithm. 
    $y$-axis: relative difference of MAE compared to using the true retention rate 0.8. 
   \label{Fig:uncertain}}
\endminipage\hfill
\minipage[t]{0.47\textwidth}
  \includegraphics[width=\linewidth]{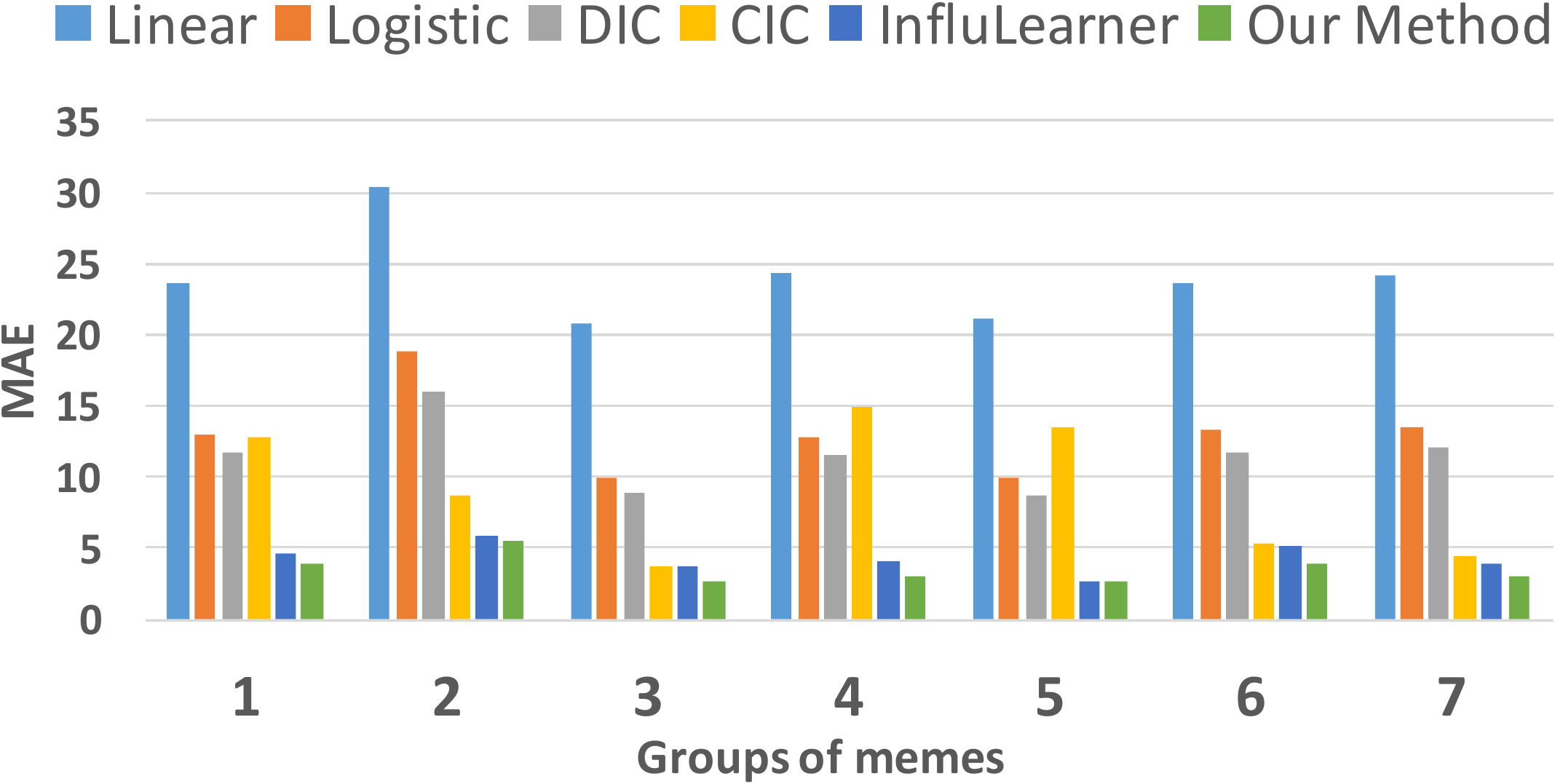}
  \caption{MAE of influence estimation on seven sets of real-world cascades 
   with 20\% of activations missing. \label{Fig:real_results}}
\endminipage\hfill
\end{figure}


To test the performance of influence function learning under
incomplete observations, we randomly delete 20\% of the
  occurrences, setting $\RetRate=0.8$. 
The results for other retention rates are similar and omitted.
Figure~\ref{Fig:real_results} shows the MAE of our methods and the
five baselines, averaged over 100 random draws of test seed sets, for
all groups of memes.
While some baselines perform very poorly, even compared to the best
baseline (InfluLearner), our algorithm provides an 18\% reduction in
MAE (averaged over the seven groups), showing the potential of data
loss awareness to mitigate its effects.

\section{Model Extensions} \label{Sec:discussion}
So far, we have assumed that the retention rate is the same for all the
nodes; however, our approach can be easily generalized to the case in
which each individual node has a different (but known) retention rate.
The following theorem generalizes Theorems~\ref{Thm:DIC} and
\ref{Thm:DLT}. 
(The proofs of all theorems from this section are given in
Appendix~\ref{app:discussion-proofs}.)

\begin{theorem}\label{Thm:DiffLossRate}
Let $\lambda\in(0, 0.5)$ and for each node $v$, let $\RetRate_v$ be
$v$'s retention rate. 
Write $\bar{r}=\frac{1}{n}\sum_{i=1}^n\frac{1}{\RetRate_v^2}$.
\begin{itemize}
\item The class of influence functions under the DIC model 
in which all edge activation probabilities satisfy
$\EdgeW[e]\in [\lambda, 1-\lambda]$ is PAC learnable 
with sample complexity 
$\tilde{O}(\frac{\bar{r}^2\NodeNum^3 \EdgeNum}{\varepsilon^2})$.
\item The class of influence functions under the DLT model 
such that the edge weight for every edge satisfies
$\EdgeW[e]\in [\lambda, 1-\lambda]$,
and for every node $v$, 
$1 - \sum_{u\in N(v)}\EdgeWD{u}{v} \in [\lambda, 1-\lambda]$, 
is PAC learnable with sample complexity 
$\tilde{O}(\frac{\bar{r}^2\NodeNum^3 \EdgeNum}{\varepsilon^2})$. 
\end{itemize}
\end{theorem}
The empirical evaluation in the previous section also shows that the
performance does not change significantly if the true retention rate of
each node is independently perturbed around the estimated mean loss
rate.

A second limitation of our approach is that we assume the retention rate 
\RetRate to be known to the algorithm. 
Estimating \RetRate in a real-world setting presents a ``chicken and
egg'' problem; however, we believe that a somewhat accurate estimate
of \RetRate (perhaps based on past data for which ground truth can be
obtained at much higher cost) will still be a significant improvement
over the status quo, namely, pretending that no data are missing.  

Moreover, even approximate information about \RetRate leads to
positive results on proper PAC learnability.
We show that the PAC learnability result can be extended to the case
where we only know that the true retention rate lies in a given interval. 
Instead of knowing the exact value \RetRate, 
we only know that \RetRate lies in an interval measured
by the relative error $\RelativeInterval$, namely $\RetRate\in I =
[\RetRateCenter\cdot(1-\RelativeInterval),
\RetRateCenter\cdot(1+\RelativeInterval)]$.
Within that interval, the retention rate is adversarially
chosen. We can then generalize Theorems~\ref{Thm:DIC} and
\ref{Thm:DLT} as follows.

\begin{theorem}\label{Thm:IntervalLossRate}
Let $\lambda\in(0, 0.5)$, and assume that the ground truth retention
rate $\RetRate$ is adversarially chosen in
$I = [\RetRateCenter\cdot(1-\RelativeInterval),
\RetRateCenter\cdot(1+\RelativeInterval)]$.
For all $\varepsilon, \delta\in (0,1)$, 
all parametrizations of the diffusion model,  
and all distributions \SeedDist over seed sets \SeedSet: 
when given \emph{activation-only} and \emph{incomplete} training cascades
$\IncomCascs$, there exists a proper learning algorithm 
$\mathcal{A}$ which outputs an influence function $\InfFunc \in
\ModelClass$ satisfying:
$$
\Prob[\IncomCascs]{\Err[\InfFunc] - \Err[\InfFunc^*] \geq \varepsilon 
+ \frac{4\RelativeInterval^2}{(1-\RelativeInterval)^2}} \leq \delta.
$$
\begin{itemize}
\item For the DIC model, when all edge activation probabilities satisfy
$\EdgeW[e]\in [\lambda, 1-\lambda]$, 
the required number of observations is
$\CascNum=\tilde{O}(\frac{\NodeNum^3 \EdgeNum}{\varepsilon^4\RetRate^4(1-\RelativeInterval)^4})$.
\item For the DLT model, when all edges weights satisfy
$\EdgeW[e]\in [\lambda, 1-\lambda]$,
and for every node $v$, 
$1 - \sum_{u\in N(v)}\EdgeWD{u}{v} \in [\lambda, 1-\lambda]$, 
the required number of observations is
$\CascNum=\tilde{O}(\frac{\NodeNum^3 \EdgeNum}{\varepsilon^4\RetRate^4(1-\RelativeInterval)^4})$.
\end{itemize}
\end{theorem}

Notice that the result of Theorem~\ref{Thm:IntervalLossRate}
is not technically a PAC learnability result,
due to the additive error that depends on the interval
size.
However, the theorem provides useful approximation guarantees when the
interval size is small.
A dependence of the guarantee on the interval size is inevitable. 
For when nothing is known about the retention rate 
(for $\RelativeInterval$ large enough), 
all information about the marginal activation
probabilities is lost in the incomplete data: 
for instance, if no nodes are ever obesrved active, we cannot
  distinguish the case $\RetRate = 0$ from the case in which no nodes
  become activated.
The experiments from the previous section confirm that for moderate
uncertainty about the retention rate, the performance of our approach is
not very sensitive to the misestimation of \RetRate.

\section{Conclusion and Future Work}
We studied the problem of learning influence functions under
incomplete observations, which are common in real-world
applications. 
We established proper PAC learnability of influence functions under
two popular diffusion models, the DIC and DLT model. 
The incompleteness only has moderate impact on the sample complexity
bound, but computational efficiency would require an oracle
for efficient empirical risk minimization. 
We next designed an efficient improper learning algorithm with
learning guarantees for the DIC, DLT, and CIC models.

Our framework can be easily generalized to handle non-uniform (but
independent) loss of node activations.
We also have partial results theoretically establishing
robustness to misestimations of \RetRate (which we observed
experimentally in Section~\ref{Sec:experiments}).
A much more significant departure for future work would be non-random
loss of activations, e.g., losing all activations of some randomly
chosen nodes.
As another direction, it would be worthwhile to generalize the PAC
learnability results to other diffusion models, and to design
an efficient algorithm with PAC learning guarantees.
\section*{Acknowledgments}
We would like to thank anonymous reviewers for useful feedback. 
The research was sponsored in part
by NSF research grant IIS-1254206
and by the U.S.~Defense Advanced Research Projects Agency 
(DARPA) under the Social Media in Strategic Communication
(SMISC) program, Agreement Number W911NF-12-1-0034. 
The views and conclusions are those of the authors and should not be
interpreted as representing the official policies of the funding agency,
or the U.S. Government.
\newpage
\bibliographystyle{abbrv}

\newpage
\appendix 
\section{Proofs for Section~\ref{Sec:theory}}\label{app:theory-proofs}
\subsection{Proof of Theorem~\ref{Thm:DIC}}
\label{App:DIC-proof}
Here, we flesh out the proof sketch from Section~\ref{Sec:theory} for
the DIC model. For the transformed graph \TransGraph, we consider only the
influence functions of the \NodeNum nodes in the added layer
\NewNodeSet.
Recall that we write 
$\NewInfFuncV{\SeedSet} = [\NewInfFuncD{1'}{\SeedSet}, \ldots, \NewInfFuncD{\NodeNum'}{\SeedSet}]$ 
for the influence function of those nodes. 
Let $\hat{\InfFunc}^{*}$
be the ground truth influence function for the same nodes, 
and $\InfFunc^{*}$ the ground truth influence function for $\Graph$. 
Let $\Model(\Graph)$ and $\Model(\TransGraph)$ be the class of
influence functions of \Graph and \TransGraph. 
For functions $\hat{\InfFunc}$, we write 
$\IncomErr[\NewInfFunc] = \Expect[\SeedSet, \ActSet]{\frac{1}{n}\sum_{v'\in\NewNodeSet}(\CharFuncD{v'}{\ActSet} - \NewInfFuncD{v'}{\SeedSet})^2}$.
Notice that the ground truth functions minimize the expected
squared error, i.e.,
$\hat{\InfFunc}^{*} \in \argmin_{\NewInfFunc\in \Model(\TransGraph)} \IncomErr[\NewInfFunc]$ 
and 
$\InfFunc^{*} \in \argmin_{\InfFunc\in \Model(\Graph)} \Err[\InfFunc]$.
We will show that $\Err[\InfFunc] - \Err[\InfFunc^{*}]$ can be made
arbitrary small.


We first prove a variation of Theorem~2 from
\cite{narasimhan_learnability_2015} 
for learning $\hat{\InfFunc}$, 
by verifying that all the supporting lemmas still apply.
The modified Theorem~2 from \cite{narasimhan_learnability_2015} is the
following:

\begin{theorem} \label{Thm:dic_transform}
Assume that the learning algorithm observes 
$\CascNum=\tilde{\Omega}(\hat{\epsilon}^{-2}\NodeNum^3 \EdgeNum)$ training cascades 
$\Transform{\Cascs} = \{(\SeedSet[i], \ActSetP[i])\}$ under the DIC model. 
Then, with probability at least $1-\delta$, we have 
\begin{eqnarray}
\IncomErr[\hat{\InfFunc}] - \IncomErr[\hat{\InfFunc}^{*}] & \leq & \hat{\epsilon}.\label{Equ:theorem2}
\end{eqnarray}
\end{theorem}

\begin{proof}
While the cascades in $\hat{\Cascs}$ are incomplete on $\NodeSet$, 
they are \emph{complete} on $\NewNodeSet$. 
We use this completeness of the cascades as follows.
Consider the restricted class of the DIC model on the transformed
graph $\TransGraph$ in which only the $\EdgeNum$ activation probabilities $\EdgeWV$ between
nodes in $\NodeSet$ are learnable, while the edges $(v,v')$ have a
fixed weight of \RetRate.
Define the log-likelihood for a cascade $(\SeedSet, \ActSetP)$ as
$$
\mathcal{L}(\SeedSet, \ActSetP|\EdgeWV) 
= \sum_{v' \in\NewNodeSet} \CharFuncD{v'}{\ActSetP[i]}\log (\NewInfFuncDU{v'}{\EdgeWV}{\SeedSet}) 
+ (1 - \CharFuncD{v'}{\ActSetP[i]})\log (1 - \NewInfFuncDU{v'}{\EdgeWV}{\SeedSet}).
$$
The algorithm outputs an influence function $\hat{\InfFunc}$
based on the solution of the following optimization problem:
$$
\EdgeWV^* \in {\argmax}_{\EdgeWV\in[\lambda, 1-\lambda]^\EdgeNum}\sum_{i=1}^{\CascNum} \mathcal{L}(\SeedSet[i], \ActSetP[i]|\EdgeWV).
$$ 

As the function $\hat{\InfFunc}$ is learned
from the DIC model, Lemma~3 in \cite{narasimhan_learnability_2015}
carries thorough to establish the Lipschitz continuity of DIC
influence functions.
\begin{lemma}[Lipschitz continuity of DIC influence]
Fix $\SeedSet \subseteq \NodeSet$ and $v' \in \NewNodeSet$. 
For any $\EdgeWV, \EdgeWV' \in \mathbb{R}^{\EdgeNum}$ with 
$||\EdgeWV - \EdgeWV'||_1 \leq \epsilon$, we have
$|\NewInfFuncDU{v'}{\EdgeWV}{\SeedSet} - \NewInfFuncDU{v'}{\EdgeWV'}{\SeedSet}| \leq \epsilon$.
\end{lemma}

Moreover, such instances (on $2n$ nodes) still only have $\EdgeNum$
parameters, and the $L_{\infty}$ covering number bound in Lemma~8
from~\cite{narasimhan_learnability_2015} applies without any changes.

\begin{lemma}[Covering number of DIC influence functions]
The $L_\infty$ covering number of the restricted class of the DIC
influence functions on the transformed graph for radius $\epsilon$ is
$O((\EdgeNum/\epsilon)^{\EdgeNum})$.
\end{lemma}

Establishing the sample complexity bound on the log-likelihood
objective (Lemma~4 in \cite{narasimhan_learnability_2015}) 
requires that all function values be bounded away from 0 and 1 
(Lemma~9 in \cite{narasimhan_learnability_2015}).
We assume that $\RetRate<1$, as Lemma~4 in
\cite{narasimhan_learnability_2015} directly holds when there are no
missing data at all. 
Let $\lambda > 0$ be the bound on the edge activation probabilities in 
\Graph from our Theorem~\ref{Thm:DIC}; 
that is, $\lambda \leq \EdgeWD{u}{v} \leq 1-\lambda$ for all
$u,v \in \NodeSet$.
Due to the layered structure of \TransGraph, 
we have that
$\RetRate \cdot \lambda^{\NodeNum} \leq \NewInfFuncD{v'}{\SeedSet}
\leq \RetRate \cdot (1- \lambda^{\NodeNum})$.\footnote{%
As in the proof of Lemma~4 in \cite{narasimhan_learnability_2015}, we
assume that there exists a path in the graph \TransGraph from a node
in \SeedSet to node $v'$; the cases where this assumption fails can
be handled easily.}
Therefore, Lemma 4 in~\cite{narasimhan_learnability_2015} carries
thorough with the same sample complexity of
$\tilde{O}(\hat{\epsilon}^{-2}\NodeNum^3 \EdgeNum)$:

\begin{lemma}[Sample complexity guarantee on the log-likelihood objective]
Fix $\epsilon, \delta\in (0, 1)$ and 
$\CascNum=\tilde{\Omega}(\hat{\epsilon}^{-2}\NodeNum^3 \EdgeNum)$. 
With probability at least $1 - \delta$ (over the draws of the training cascades), 
$$
\max_{\EdgeWV\in[\lambda, 1-\lambda]^\EdgeNum}
  \Expect[\SeedSet,\ActSetP]{\frac{1}{\NodeNum}
    \mathcal{L}(\SeedSet, \ActSetP | \EdgeWV)} 
- \Expect[\SeedSet,\ActSetP]{\frac{1}{\NodeNum}
    \mathcal{L}(\SeedSet, \ActSetP|\EdgeWV^*)} 
\leq \epsilon. 
$$
\end{lemma}

As all the lemmas used in the proof of Theorem 2
from~\cite{narasimhan_learnability_2015} remain true, 
we have proved our Theorem~\ref{Thm:dic_transform}, with the 
guarantee that $\Err[\hat{\InfFunc}] - \Err[\hat{\InfFunc}^{*}] \leq \hat{\epsilon}$.
\end{proof}

Finally, we recall that according to Equation~\eqref{Equ:tranformation},
$\InfFuncD{v}{\SeedSet} = \frac{1}{\RetRate} \cdot \NewInfFuncD{v'}{\SeedSet}$
and 
$\InfFuncDU{v}{*}{\SeedSet} = \frac{1}{\RetRate} \cdot \NewInfFuncDU{v'}{*}{\SeedSet}$,
giving us that

\begin{eqnarray}
\Err[\InfFunc] - \Err[\InfFunc^{*}] 
& \stackrel{(*)}{=} & 
\frac{1}{n}\sum_{v\in \NodeSet}
  \Expect[\SeedSet]{(\InfFuncD{v}{\SeedSet} - \InfFuncDU{v}{*}{\SeedSet})^2} \nonumber\\
& \stackrel{\text{Equation~\eqref{Equ:tranformation}}}{=} & 
\frac{1}{n}\sum_{v'\in \TransNodeSet}
  \Expect[\SeedSet]{(\frac{1}{\RetRate}\NewInfFuncD{v'}{\SeedSet} - \frac{1}{\RetRate}\NewInfFuncDU{v'}{*}{\SeedSet})^2} \nonumber\\
& \stackrel{(*)}{=} & 
\frac{\IncomErr[\hat{\InfFunc}] - \IncomErr[\hat{\InfFunc}^{*}]}{\RetRate^2} \nonumber\\
& \stackrel{\text{Equation~\eqref{Equ:theorem2}}}{\leq} & 
\frac{\hat{\epsilon}}{\RetRate^2} \nonumber
\end{eqnarray}
(The steps labeled (*) are applications of Equation~(4) from \cite{narasimhan_learnability_2015}.)
Now, by taking 
$\hat{\epsilon} = \varepsilon\cdot\RetRate^2$, 
with $\tilde{O}(\frac{\NodeNum^3 \EdgeNum}{\varepsilon^2\RetRate^4})$
incomplete cascades, we obtain that 
$\Err[\InfFunc] - \Err[\InfFunc^{*}] \leq \varepsilon$.
\subsection{Proof of Theorem~\ref{Thm:DLT}}
We will show that the analogue of  Theorem~\ref{Thm:dic_transform} for
the DLT model also holds. 
We do so by following the same sequence of steps as in
Appendix~\ref{App:DIC-proof} and verifying that all the steps in the
proof of Theorem~2 in \cite{narasimhan_learnability_2015} still hold.
The main difference is that a new proof is needed for establishing
Lipschitz continuity of the DLT influence function with respect to the
$L_1$ norm (the analogue of Lemma 3 in \cite{narasimhan_learnability_2015}).
We begin by establishing this lemma.

\begin{lemma}[Lipschitz continuity] \label{Lem:lipschtiz}
Fix $\SeedSet \subseteq \NodeSet$ and $u \in \NodeSet$.
For any $\EdgeWV, \EdgeWV' \in \mathbb{R}^{\EdgeNum}$ 
with $||\EdgeWV - \EdgeWV'||_1 \leq \varepsilon$, 
we have that
$|\InfFuncDU{u}{\EdgeWV}{\SeedSet} - \InfFuncDU{u}{\EdgeWV'}{\SeedSet}| \leq \varepsilon$.
\end{lemma}

\begin{proof}
As shown in~\cite{kempe_maximizing_2003}, the influence functions under
the DLT model can be also characterized via the reachability under a
distribution over live-edge graphs. 
Specifically, the distribution is as follows 
\cite[Theorem~2.5]{kempe_maximizing_2003}: 
for each node $v$, pick at most one of its incoming edges at random,
selecting the edge from $z \in N(v)$ with probability 
\EdgeWD{z}{v} and selecting no incoming edge with probability 
$1 - \sum_{z \in N(v)} \EdgeWD{z}{v}$. 
For each node $v$, let the random variable $X_v$ be the
incoming neighbor chosen for $v$, with $X_v = \perp$ if $v$ has no
incoming edge. 
For simplicity of notation, we define
$\EdgeWD{\perp}{v} = 1 - \sum_{z \in N(v)} \EdgeWD{z}{v}$.
Define $\bm{X} = (X_v)_{v \in \NodeSet}$,
and write $\mathcal{X}$ for the set of all such vectors $\bm{X}$.
For any node $v$, we write $\mathcal{X}_{-v}$ for the set of all
vectors with edges (or $\perp$) for all nodes except $v$.
And for a vector $\bm{X} \in \mathcal{X}_{-v}$, we write
$\bm{X}[v\mapsto z]$ for the vector in which all entries agree with
those in $\bm{X}$, except for the entry for $v$ which is now $z$.

Let $R_{\bm{X}}(v, \SeedSet)$ be the indicator function of whether node
$v$ is reachable from the seed set \SeedSet in the graph
$(V,\bm{X})$, where we interpret $\bm{X}$ as the set of all
edges $(X_v,v)$ with $X_v \neq \perp$.
Claim~2.6 of \cite{kempe_maximizing_2003} implies that
$$
\InfFuncDU{u}{\EdgeWV}{\SeedSet} 
= \sum_{\bm{X}} \prod_{v\in\NodeSet}{\EdgeWD{X_v}{v}} R_{\bm{X}} (u, \SeedSet).
$$

We fix an edge $(y,y')$ and take the partial derivative of 
$\InfFuncDU{u}{\EdgeWV}{\SeedSet}$ with respect to \EdgeWD{y}{y'}:
\begin{eqnarray}
\left|\frac{\partial \InfFuncDU{u}{\EdgeWV}{\SeedSet}}{\partial \EdgeWD{y}{y'}}\right| 
& = & \left|\frac{\partial}{\partial \EdgeWD{y}{y'}} 
  \left( \sum_{z \in N(y) \cup \{\perp\}} \EdgeWD{z}{y} 
    \sum_{\bm{X} \in \mathcal{X}_{-y}}
        \prod_{v \in \NodeSet \setminus \{y\}} \EdgeWD{X_v}{v}
        \cdot R_{\bm{X}[y \mapsto z]}(u, \SeedSet) \right) \right|
    \nonumber\\ 
& = & \left| \sum_{\bm{X} \in \mathcal{X}_{-y}}
        \prod_{v \in \NodeSet \setminus \{y\}} \EdgeWD{X_v}{v}
             \cdot R_{\bm{X}[y \mapsto y']}(u, \SeedSet)
    - \sum_{\bm{X} \in \mathcal{X}_{-y}}
          \prod_{v \in \NodeSet \setminus \{y\}} \EdgeWD{X_v}{v}
             \cdot R_{\bm{X}[y \mapsto \perp]}(u, \SeedSet) \right|
    \nonumber\\ 
& \leq & \left| \sum_{\bm{X} \in \mathcal{X}_{-y}}
        \prod_{v \in \NodeSet \setminus \{y\}} \EdgeWD{X_v}{v} \right|
     \nonumber\\ 
& = & 1. \nonumber
\end{eqnarray}
Therefore, 
$||\nabla_{\EdgeWV} \InfFuncDU{u}{\EdgeWV}{\SeedSet}||_\infty \leq 1$,
implying Lipschitz continuity.
\end{proof}

Next, we bound the values of the influence functions away from 0 and 1.
Because each edge weight $\EdgeWD{z}{v} \in [\lambda,1-\lambda]$ by
assumption, and we further assumed that
$\EdgeWD{\perp}{v} = 1 - \sum_{z \in N(v)} \EdgeWD{z}{v} \in [\lambda, 1-\lambda]$, 
it follows directly (as in the proof for the DIC model) that
$\RetRate \cdot \lambda^{\NodeNum} 
\leq \NewInfFuncD{v'}{\SeedSet} 
\leq \xhedit{\RetRate} \cdot (1-\lambda^{\NodeNum})$.
This establishes the analogue of Lemma~9 in
\cite{narasimhan_learnability_2015}, 
and we can therefore apply Lemma~4 in \cite{narasimhan_learnability_2015},
obtaining a sample complexity of
$\tilde{O}(\hat{\epsilon}^{-2}\NodeNum^3 \EdgeNum)$ under the DLT model. 
As all the used lemmas remain true, the results of
Theorem~\ref{Thm:dic_transform} also hold for the DLT model. 
Finally, exploiting the same relation between $\InfFuncV{\SeedSet}$
and $\NewInfFuncV{\SeedSet}$ as in the proof of Theorem~\ref{Thm:DIC}
leads to the conclusion of Theorem~\ref{Thm:DLT}.

\section{Proof of Theorem~\ref{Thm:empirical}}
\label{app:algorithm-proof}
Let $\CascNum=\tilde{\Omega}(\frac{\log \FeatureConst}{\epsilon^4\RetRate^2})$,
and let \InfFuncDU{v}{\tilde{\ParamV}, \TrunConst}{\SeedSet} be the
influence functions obtained in Theorem~\ref{Thm:empirical}.
We will show that for any single node $v$, 
with probability at least $1 - \delta/\NodeNum$,
$$
\Expect[\SeedSet]{ (\InfFuncDU{v}{\tilde{\ParamV}, \TrunConst}{\SeedSet} 
                    - \InfFuncDU{v}{*}{\SeedSet})^2} \leq \varepsilon.
$$
The theorem then follows by taking a union bound over all \NodeNum nodes.

Recall that \Model[\TrunConst] is function class of all truncated
influence functions.
We write 
$$
\Rade[\CascNum]{\Model[\TrunConst]} := 
\Expect[{\SeedSet[i]}\sim\SeedDist, (\epsilon_i)_i \sim \text{Uniform}(\{-1,1\}^{\CascNum})]{
\sup_{\InfFunc\in\Model[\TrunConst]} \frac{1}{\CascNum}
\sum_{i=1}^{\CascNum} \epsilon_i \cdot \InfFuncD{v}{\SeedSet[i]}}
$$
for its Rademacher complexity, 
where the $\epsilon_i$'s are i.i.d.~Rademacher (symmetric Bernoulli)
random variables. 
By Lemma~12 in \cite{du_influence_2014}, 
there exists a truncated influence function
$\InfFuncDUO{v}{\hat{\ParamV}, \TrunConst} \in \Model[\TrunConst]$ 
with
$\FeatureNum = O(\frac{\FeatureConst^2}{\varepsilon^2} \log \frac{\FeatureConst\NodeNum}{\varepsilon\delta})$
features such that 
$\Expect[\SeedSet \sim \SeedDist]{
  (\InfFuncDU{v}{\hat{\ParamV}, \TrunConst}{\SeedSet} -
  \InfFuncDU{v}{*}{\SeedSet})^2 } \leq 2\varepsilon^2 + 2\TrunConst^2$
with probability at least $1-\frac{\delta}{2n}$. 

Using the log likelihood function 
$\LossFunc(t, \Label) = \Label \log t + (1-\Label)\log(1-t)$
as defined in Section~\ref{Sec:algorithm}, we write the log loss of
influence function \InfFunc[v] as 
$\ErrLog[\InfFunc[v]] = \Expect[\SeedSet, \ActSet]{
-\LossFunc(\InfFuncD{v}{\SeedSet},\ActSet)}$.
By Lemma~2 in \cite{natarajan2013learning}, 
with probability at least $1-\frac{\delta}{2\NodeNum}$,
$$
\ErrLog[\InfFuncDUO{v}{\tilde{\ParamV}, \TrunConst}] 
\; \leq \; \min_{f \in \Model[\TrunConst]} \ErrLog[f] 
  + \frac{4}{\TrunConst \cdot \RetRate} \Rade[\CascNum]{\Model[\TrunConst]}
  + \sqrt{\frac{\log(2\NodeNum/\delta)}{2\CascNum}}.
$$ 

Because $\InfFuncDUO{v}{\hat{\ParamV}, \TrunConst} \in \Model[\TrunConst]$,
we can bound that
$\min_{f \in \Model[\TrunConst]} \ErrLog[f] \leq 
\ErrLog[\InfFuncDU{v}{\hat{\ParamV}, \TrunConst}{\SeedSet}]$
on the right-hand side.
Subtracting \ErrLog[\InfFuncDUO{v}{*}] from both sides, we obtain
\begin{eqnarray}
\ErrLog[\InfFuncDUO{v}{\tilde{\ParamV}, \TrunConst}] 
- \ErrLog[\InfFuncDUO{v}{*}] 
& \leq &
\ErrLog[\InfFuncDUO{v}{\hat{\ParamV}, \TrunConst}] 
- \ErrLog[\InfFuncDUO{v}{*}] 
+ \frac{4}{\TrunConst \cdot \RetRate} \Rade[\CascNum]{\Model_{\TrunConst}}
+ \sqrt{\frac{\log(2\NodeNum/\delta)}{2\CascNum}}. \label{Eqn:loss-upper-bound}
\end{eqnarray}

The square and log errors can be related to each other as in
  the proof of Theorem 2 in~\cite{narasimhan_learnability_2015}, as follows:
$$
\Expect[\SeedSet]{ (\InfFuncDU{v}{\tilde{\ParamV}, \TrunConst}{\SeedSet} 
   - \InfFuncDU{v}{*}{\SeedSet})^2 } 
\; \leq \; 
\frac{1}{2}(\ErrLog[\InfFuncDUO{v}{\tilde{\ParamV}, \TrunConst}] - \ErrLog[\InfFuncDUO{v}{*}]).
$$
Hence, in order to obtain a bound on
$\Expect[\SeedSet]{ (\InfFuncDU{v}{\tilde{\ParamV}, \TrunConst}{\SeedSet} 
   - \InfFuncDU{v}{*}{\SeedSet})^2 }$,
it suffices to upper-bound the right-hand side of \eqref{Eqn:loss-upper-bound}.
The term $\ErrLog[\InfFuncDUO{v}{\hat{\ParamV}, \TrunConst}] 
- \ErrLog[\InfFuncDUO{v}{*}]$ can be bounded as in the proof of Lemma~2
in \cite{du_influence_2014}, using Lemma~11 and Lemma~16
from \cite{du_influence_2014}:
Assume that $\InfFuncDUO{v}{\hat{\ParamV}, \TrunConst}$ uses 
$K=\Omega(\frac{\FeatureConst^2}{\hat{\epsilon}^2} \log \frac{\FeatureConst\NodeNum}{\hat{\epsilon}\hat{\delta}})$
features. 
Then, with probability at least $1-\hat{\delta}$, we have
\begin{eqnarray}
\ErrLog[\InfFuncDUO{v}{\hat{\ParamV}, \TrunConst}] - \ErrLog[\InfFuncDUO{v}{*}] 
& \leq & \frac{\hat{\epsilon}^2 + \TrunConst^2}{\TrunConst}(1 + \log\frac{1}{\TrunConst}).
\label{Eqn:error-difference-bound}
\end{eqnarray}

Next, we bound the Rademacher complexity of the function class \Model[\TrunConst]:
\begin{lemma}
The Rademacher complexity \Rade[\CascNum]{\Model[\TrunConst]} for the
function class \Model[\TrunConst] with at most \FeatureNum features 
is at most $\sqrt{\frac{2\log (1 + \FeatureNum)}{\CascNum}}$.
\end{lemma}

\begin{proof}
Recall that we use basis functions
$\BaseFuncArg[i]{\SeedSet} := \min\{1, \CharFuncV{\SeedSet}^{\top} \FeatureVD{T_i}\}$. 
Let $\mathcal{W} = \{\BaseFunc[i] | i=1,\ldots,\FeatureNum\}\cup\{\mathbbm{1}\}$, 
where $\mathbbm{1}$ is the constant function with value $1$.
By definition, we have $\Model_{\TrunConst} \subseteq \conv{\mathcal{W}}$, 
where $\conv{\mathcal{W}}$ denotes the convex hull. 
Therefore, 
$\Rade[\CascNum]{\Model[\TrunConst]} 
\leq \Rade[\CascNum]{\conv{\mathcal{W}}}
= \Rade[\CascNum]{\mathcal{W}}$. 
Since $|\BaseFuncArg[i]{\SeedSet}| \leq 1$, 
by Massart's finite lemma\footnote{%
Massart's finite lemma states the following: 
Let $\mathcal{F}$ be a finite set of functions, such that
$\sup_{f \in \mathcal{F}} \frac{1}{n} \sum_{i=1}^n f(X_i)^2 \leq C^2$
for any variables values $X_1, \ldots, X_n$.
Then, the Rademacher complexity of $\mathcal{F}$ is upper bounded by 
$\Rade[n]{\mathcal{F}} \leq \sqrt{\frac{2C^2\log|\mathcal{F}|}{n}}$.},
we have $\Rade[\CascNum]{\mathcal{W}} \leq \sqrt{\frac{2\log (1 + \FeatureNum)}{\CascNum}}$,
completing the proof.
\end{proof} 

To finish the proof of Theorem~\ref{Thm:empirical}, 
let $\epsilon$ be the desired accuracy. 
Define $\hat{\delta}=\frac{\delta}{2\NodeNum}$ and
$\hat{\epsilon} = \TrunConst = \frac{\epsilon}{c'\log\frac{1}{\epsilon}}$, 
where $c'$ is a sufficiently large constant. 
Then, the right-hand side of \eqref{Eqn:error-difference-bound} is
upper-bounded by 
$\hat{\epsilon} \cdot (1 + \log\frac{1}{\hat{\epsilon}}) \leq \frac{\epsilon}{2}$.

With $\CascNum = \Omega(\frac{\log K}{\epsilon^4 \RetRate^2})$, 
we have
$\frac{4}{\TrunConst \cdot \RetRate}\Rade[\CascNum]{\Model[\TrunConst]}
\leq \frac{\epsilon}{4}$. 
Whenever
$\CascNum = \Omega(\frac{\log\frac{n}{\delta}}{\epsilon^2})$, we also
get that 
$\sqrt{\frac{\log(n/\delta)}{2\CascNum}} \leq \frac{\epsilon}{4}$. 
Taking \CascNum as the maximum of the above three, 
which is satisfied when 
$\CascNum = \tilde{\Omega}(\frac{\log \FeatureConst}{\epsilon^4 \RetRate^2})$,
we can substitute all of the bounds into 
the right-hand side of \eqref{Eqn:loss-upper-bound} and obtain that
$\Expect[\SeedSet]{ (\InfFuncDU{v}{\tilde{\ParamV}, \TrunConst}{\SeedSet} 
                   - \InfFuncDU{v}{*}{\SeedSet})^2 } \leq \epsilon$ 
with probability at least $1 - \frac{\delta}{\NodeNum}$. 
Now, taking a union bound over all nodes $v$ concludes the proof. 

\paragraph{Discussion.} 
Notice that when the retention rate is 1, 
our Theorem~\ref{Thm:empirical} significantly improves the sample
complexity bound compared to Du et al.~\cite{du_influence_2014}.
The sample complexity in~\cite{du_influence_2014} is
$\tilde{O}(\frac{C^2}{\epsilon^3})$,
while our theorem implies a sample complexity of
$\tilde{O}(\frac{\log C}{\epsilon^4})$ under complete observations.
The improvement is derived from bounding the Rademacher complexity 
of the function class \Model[\TrunConst] instead of the
$L_{2,\infty}$ dimension. 
The Rademacher bound leads to a logarithmic dependence of the sample
complexity on the number of features \FeatureNum, whereas the 
$L_{2,\infty}$ bound results in a polynomial dependence.

\section{Proofs for Section~\ref{Sec:discussion}}
\label{app:discussion-proofs}
\subsection{Proof of Theorem~\ref{Thm:DiffLossRate}}
{%
Let $\RetRate_i$ be the retention rate of node $i$ and 
$\hat{\varepsilon}_i$ the desired error guarantee for learning the
influence function $F_{v_i}$.
It is immediate from the proofs of Theorems~\ref{Thm:DIC} and
\ref{Thm:DLT} that 
$M=\max_i \{\tilde{O}(\frac{\NodeNum^3
  \EdgeNum}{\hat{\varepsilon}_i^2 \RetRate_i^4})\}$ incomplete
cascades are sufficient to guarantee that with probability at least
$1-\delta$, for each $i$, we obtain
$$
  \Expect[\SeedSet]{(\InfFuncD{v_i}{\SeedSet} - \InfFuncDU{v_i}{*}{\SeedSet})^2}\leq \hat{\varepsilon}_i.
$$

The estimation error for the overall influence is the average
$\varepsilon = \frac{1}{\NodeNum} \sum_i \hat{\varepsilon_i}$.
Given non-uniform retention rates $\RetRate_i$, we can choose non-uniform
$\hat{\varepsilon}_i$ yielding the desired $\varepsilon$, so as to
minimize the sample complexity. 
The corresponding optimization problem is the following:
\begin{eqnarray*}
\mbox{Minimize} & &
\max_i \frac{1}{\hat{\varepsilon}_i^2 \RetRate_i^4}\\
\mbox{subject to} & & 
\frac{1}{\NodeNum}\sum_i \hat{\varepsilon}_i = \varepsilon, \\
&& \hat{\varepsilon}_i>0\ \mbox{ for all } i. 
\end{eqnarray*}
The minimum is achieved when all
$\frac{1}{\hat{\varepsilon}_i^2 \RetRate_i^4}$ are equal to some
constant $C$, meaning that 
$\hat{\varepsilon}_i = \frac{1}{\sqrt{C} \cdot \RetRate_i^2}$.
The constant $C$ can be obtained from the constraint 
$\frac{1}{\NodeNum}\sum_i \hat{\varepsilon}_i = \varepsilon$, 
yielding that 
$C=\frac{\bar{r}^2}{\varepsilon^2}$, 
where $\bar{r}=\frac{1}{n}\sum_i \frac{1}{\RetRate_i^2}$.
This completes the proof of the theorem.

\subsection{Proof of Theorem~\ref{Thm:IntervalLossRate}}
The proof of Theorem~\ref{Thm:IntervalLossRate} is similar to that of
Theorems~\ref{Thm:DIC} and \ref{Thm:DLT}. 
We again treat the incomplete cascades as complete cascades in a
transformed graph $\TransGraph$. 
Because we no longer know the true retention rate \RetRate,
we cannot set the probability on the egde $(v,v')$ to
$\TransEdgeWD{v}{v'} = \RetRate$. 
Instead, we treat the $\TransEdgeWD{v}{v'}$ as parameters to infer,
under the constraint that 
$\TransEdgeWD{v}{v'}\in I =
[\RetRateCenter \cdot (1-\RelativeInterval),
\RetRateCenter \cdot (1+\RelativeInterval)]$. 
We spell out the details of the proof for the DIC model; the proof for
the DLT model is practically identical.

As in Theorem~\ref{Thm:DIC}, we consider only the influence functions
of the $\NodeNum$ nodes in the added layer $\NewNodeSet$. 
Following the proof of Theorem~\ref{Thm:DIC}, 
with probability as least $1-\delta$,
using $M=\tilde{O}(\frac{\NodeNum^3 \EdgeNum}{\varepsilon^2})$ cascades, 
for all $v'\in\NewNodeSet$, 
\begin{equation}
\Expect[\SeedSet]{
(\NewInfFuncD{v'}{\SeedSet} - \NewInfFuncDU{v'}{*}{\SeedSet})^2}
\leq \hat{\varepsilon}.
\label{eqn:second-level}
\end{equation}

In the proof of Theorem~\ref{Thm:DIC}, the fact that 
$\TransEdgeWD{v}{v'} = \RetRate$ allowed us to obtain
the influence function at $v$ via
$\InfFuncD{v}{\SeedSet} = \frac{1}{\RetRate} \cdot \NewInfFuncD{v'}{\SeedSet}$.
Since the edge probabilities $\TransEdgeWD{v}{v'}$ are now inferred, 
we instead use the inferred probabilities for obtaining the activation
functions for nodes $v$.
On the other hand, the ground truth influence functions for $v$ and
$v'$ are still related via the correct value \RetRate.
Writing $\hat{\RetRate}_v = \TransEdgeWD{v}{v'}$,
this gives us the following:

\begin{align}
\InfFuncD{v}{\SeedSet}
& = \frac{\NewInfFuncD{v'}{\SeedSet}}{\hat{\RetRate}_v}
& \InfFuncDU{v}{*}{\SeedSet}
& = \frac{\NewInfFuncDU{v'}{*}{\SeedSet}}{\RetRate}.
\label{Eqn:intervalF} 
\end{align}
Consider $\Expect[\SeedSet]{(\InfFuncD{v}{\SeedSet} - \InfFuncDU{v}{*}{\SeedSet})^2}$ for any node $v$.
The expected squared estimation error for node $v$ can now be written
as follows:

\begin{eqnarray}
\Expect[\SeedSet]{(\InfFuncD{v}{\SeedSet} - \InfFuncDU{v}{*}{\SeedSet})^2}
& \stackrel{\text{Equation~\eqref{Eqn:intervalF}}}{=} & 
\Expect[\SeedSet]{\left(\frac{\NewInfFuncD{v'}{\SeedSet}}{\hat{\RetRate}_v} - \frac{\NewInfFuncDU{v'}{*}{\SeedSet}}{\RetRate}\right)^2} \nonumber\\
& = & 
\Expect[\SeedSet]{\left(\frac{\NewInfFuncD{v'}{\SeedSet} - \NewInfFuncDU{v'}{*}{\SeedSet}}{\RetRate}
+\NewInfFuncD{v'}{\SeedSet}\frac{\RetRate - \hat{\RetRate}_v}{\RetRate
  \cdot \hat{\RetRate}_v}\right)^2}\nonumber\\
& \leq & \phantom{+}
\frac{1}{\RetRate^2} \cdot 
\Expect[\SeedSet]{(\NewInfFuncD{v_i}{\SeedSet} - \NewInfFuncDU{v_i}{*}{\SeedSet})^2}
\label{Equ:term1}\\
& & + \frac{2}{\RetRate} \cdot 
\Expect[\SeedSet]{\frac{\NewInfFuncD{v'}{\SeedSet}}{\hat{\RetRate}_v}
\cdot \frac{|\RetRate - \hat{\RetRate}_v|}{\RetRate}
\cdot |\NewInfFuncD{v'}{\SeedSet} - \NewInfFuncDU{v'}{*}{\SeedSet}|}
\label{Equ:term2}\\
& & + 
\Expect[\SeedSet]{\frac{\NewInfFuncDU{v'}{2}{\SeedSet}}{\hat{\RetRate_v^2}}
\cdot \frac{(\RetRate - \hat{\RetRate}_v)^2}{\RetRate^2}}.
\label{Equ:term3}
\end{eqnarray}
We will bound the term~\eqref{Equ:term1} using
Inequality~\eqref{eqn:second-level}.
In order to bound the terms \eqref{Equ:term2} and \eqref{Equ:term3},
observe the following:
\begin{itemize}
\item For all seed sets $\SeedSet$ and nodes $v'\in\NewNodeSet$,
  we have $\NewInfFuncD{v'}{\SeedSet}\leq \hat{\RetRate}_v$ 
  by the structure of the transformed graph.
\item $\frac{|\RetRate - \hat{\RetRate}_v|}{\RetRate}\leq
  \frac{2\RelativeInterval}{1-\RelativeInterval}$ follows from the assumption that 
  $\hat{\RetRate}_v, \RetRate \in[\RetRateCenter\cdot(1-\RelativeInterval),
  \RetRateCenter\cdot(1+\RelativeInterval)]$. 
\item By Jensen's inequality and Inequality~\eqref{eqn:second-level}, 
  $\Expect[\SeedSet]{|\NewInfFuncD{v'}{\SeedSet} -
    \NewInfFuncDU{v'}{*}{\SeedSet}|}\leq \sqrt{\hat{\varepsilon}}$. 
\end{itemize}
Using the preceding three inequalities, we can bound the 
term~\eqref{Equ:term2} by
$\frac{4\RelativeInterval\sqrt{\hat{\varepsilon}}}{\RetRate(1-\RelativeInterval)}$
and the term~\eqref{Equ:term3}
by $\frac{4\RelativeInterval^2}{(1-\RelativeInterval)^2}$.

When $\hat{\varepsilon} \leq \frac{\varepsilon\RetRate^2}{2}$, 
using the inequality~\eqref{eqn:second-level}, the
  term~\eqref{Equ:term1} is upper-bounded by  
$\frac{1}{\RetRate^2}\Expect[\SeedSet]{(\InfFuncD{v_i}{\SeedSet}
- \InfFuncDU{v_i}{*}{\SeedSet})^2} \leq \frac{\varepsilon}{2}$. 
Similarly, when $\hat{\varepsilon} \leq \frac{\varepsilon^2\RetRate^2(1-\RelativeInterval)^2}{64\RelativeInterval^2}$, 
the additive term~\eqref{Equ:term2} is upper-bounded by 
$\frac{2\RelativeInterval\sqrt{\hat{\varepsilon}}}{\RetRate(1-\RelativeInterval)} 
\leq \frac{\varepsilon}{2}$. 
Thus, taking $\hat{\varepsilon}=\min\{\frac{\varepsilon\RetRate^2}{2}, \frac{\varepsilon^2\RetRate^2(1-\RelativeInterval)^2}{64\RelativeInterval^2}\}$, 
the first two terms combined are bounded by $\varepsilon$.
Thus, using
$\CascNum=\tilde{O}(\frac{\NodeNum^3 \EdgeNum}{\varepsilon^4\RetRate^4(1-\RelativeInterval)^4})$
cascades, with probability at least $1-\delta$,
for each node $v$,  
\[
\Expect[\SeedSet]{(\InfFuncD{v}{\SeedSet} - \InfFuncDU{v}{*}{\SeedSet})^2} \leq \varepsilon + \frac{4\RelativeInterval^2}{(1-\RelativeInterval)^2}.
\]
Now, taking an average on both sides of the above equation over all
the nodes $v\in\NodeSet$ concludes the proof of
Theorem~\ref{Thm:IntervalLossRate}.

\end{document}